\documentclass[graybox]{svmult}

\usepackage{mathptmx}       % selects Times Roman as basic font
\usepackage{helvet}         % selects Helvetica as sans-serif font
\usepackage{courier}        % selects Courier as typewriter font
\usepackage{type1cm}        % activate if the above 3 fonts are
\usepackage{makeidx}         % allows index generation
\usepackage{graphicx}        % standard LaTeX graphics tool
\usepackage{multicol}        % used for the two-column index
\usepackage[bottom]{footmisc}% places footnotes at page bottom

\usepackage{amsmath}
\usepackage{amssymb}
\usepackage{pxfonts}
\usepackage{graphicx}
\usepackage{eucal}
\usepackage{mathrsfs}
\usepackage{theorem}
\usepackage{pifont}
\usepackage{sectsty}
\usepackage{amscd}
\usepackage{color}
\usepackage{fancyhdr}
\usepackage{framed}
\usepackage{mathbbol}
\usepackage{mdframed}
\usepackage{enumitem}
\usepackage[all]{xy}
\usepackage[pagebackref]{hyperref}

\usepackage{xcolor} % Required for specifying colors by name
\definecolor{ocre}{RGB}{243,102,25} 
\definecolor{darkocre}{RGB}{121,51,12} 
\definecolor{lightocre}{RGB}{255,150,37} 
\definecolor{verylightocre}{RGB}{255,204,50} 
\definecolor{soccerfield}{RGB}{107,142,35} 
\definecolor{lightgray}{RGB}{200,200,200} 
\definecolor{warmblue}{RGB}{51,102,153} 
\definecolor{lightwarmblue}{RGB}{105,141,198} 
\definecolor{sepia}{RGB}{112,66,20} 

\usepackage{tikz} % Required for drawing custom shapes
\usepackage{tkz-euclide}
\usetkzobj{all}
\usepackage{pgfplots}
\usetikzlibrary{decorations.pathreplacing,decorations.markings}

\tikzset{
  on each segment/.style={
    decorate,
    decoration={
      show path construction,
      moveto code={},
      lineto code={
        \path [#1]
        (\tikzinputsegmentfirst) -- (\tikzinputsegmentlast);
      },
      curveto code={
        \path [#1] (\tikzinputsegmentfirst)
        .. controls
        (\tikzinputsegmentsupporta) and (\tikzinputsegmentsupportb)
        ..
        (\tikzinputsegmentlast);
      },
      closepath code={
        \path [#1]
        (\tikzinputsegmentfirst) -- (\tikzinputsegmentlast);
      },
    },
  },
  mid arrow/.style={postaction={decorate,decoration={
        markings,
        mark=at position .5 with {\arrow[#1]{stealth}}
      }}},
}

\newcommand{\M}{\mathcal{M}}
\newcommand{\R}{\mathbb{R}}

\newcommand{\calW}{\mathcal{W}}
\newcommand{\g}{\mathfrak{g}}
\newcommand{\h}{\mathfrak{h}}
\newcommand{\euc}{\mathfrak{e}}
\newcommand{\dVolg}{\text{\textup{dVol}}_\g}
\newcommand{\dVol}{\text{\textup{dVol}}}
\newcommand{\dVolgn}{\text{\textup{dVol}}_{\g_n}}

\providecommand{\vp}{\varphi}
\newcommand{\tI}{\tilde{I}}
\newcommand{\pl}{\partial}
\renewcommand{\div}{\operatorname{div}}
\renewcommand{\b}{\textup{\textbf{b}}}
\newcommand{\e}{\varepsilon}

\providecommand{\dist}{\operatorname{dist}}
\providecommand{\SO}{\operatorname{SO}}
\providecommand{\GL}{\operatorname{GL}}
\providecommand{\Hom}{\operatorname{Hom}}

\newcommand{\nablaLC}{\nabla^{LC}}

\newcommand{\beq}{\begin{equation}}
\newcommand{\eeq}{\end{equation}}

\newcommand{\brk}[1]{\left(#1\right)}          % \brk{.}     => (.)
\begin{document}

\title*{Limits of distributed dislocations in geometric and constitutive paradigms}
\author{Marcelo Epstein, Raz Kupferman and Cy Maor}
\institute{Marcelo Epstein \at University of Calgary \email{mepstein@ucalgary.ca}
\and Raz Kupferman \at The Hebrew University of Jerusalem \email{raz@math.huji.ac.il}
\and Cy Maor \at The Hebrew University of Jerusalem \email{cy.maor@mail.huji.ac.il}}
\maketitle

\abstract{
The 1950's foundational literature on rational mechanics exhibits two somewhat distinct paradigms to the representation of continuous distributions of defects in solids.
In one paradigm, the fundamental objects are geometric structures on the body manifold, e.g., an affine connection and a Riemannian metric, which represent its internal microstructure.
In the other paradigm, the fundamental object is the constitutive relation; if the constitutive relations satisfy a property of material uniformity,  then it induces certain geometric structures on the manifold.
In this paper, we first review these paradigms, and show that they are equivalent if the constitutive model has a discrete symmetry group (otherwise, they are still consistent, however the geometric paradigm contains more information). 
We then consider bodies with continuously-distributed edge dislocations, and show, in both paradigms, how they can be obtained as homogenization limits of bodies with finitely-many dislocations as the number of dislocations tends to infinity.
Homogenization in the geometric paradigm amounts to a convergence of manifolds; in the constitutive paradigm it amounts to a $\Gamma$-convergence of energy functionals.
We show that these two homogenization theories are consistent, 
and even identical in the case of constitutive relations having discrete symmetries.
}

\section{Introduction}

\subsection{Geometric and constitutive paradigms}

\paragraph{Geometric paradigm: body manifolds}
The 1950's foundational literature on rational mechanics exhibits two somewhat distinct paradigms to 
the representation of continuous distributions of defects in solids.
On the one hand, there is a paradigm promoted by Kondo \cite{Kon55}, Nye \cite{Nye53}, Bilby \cite{BBS55} and later Kr\"oner (e.g.~\cite{Kro81}), in which solid bodies are modeled as geometric objects---manifolds---and their internal microstructure is represented by sections of fiber bundles, such as a metric and an affine connection. 

More specifically, in \cite{Kon55,Nye53,BBS55}, the body manifold is assumed to be a smooth manifold $\M$, endowed with a notion of distant parallelism, which amounts to defining a curvature-free affine connection $\nabla$. The connection is generally non-symmetric, and its torsion tensor is associated with the density of dislocations. This geometric model is motivated by an analysis of Burgers circuits, which in the presence of dislocations exhibit geodesic rectangles whose opposite sides are not of equal lengths---a signature of torsion (see Section~\ref{sec:homogenization_geo} for a discussion of Burgers circuits and Burgers vectors in this setting).

Note that modulo the choice of a basis at a single point, the definition of a distant parallelism is equivalent  to a choice of a basis for the tangent bundle at each point (i.e., a global smooth section of the frame bundle). Intuitively, the frame field at each point corresponds to the crystalline axes one would observe under a microscope. Torsion is a measure for how those local bases twist when moving from one point to another. 

The choice of local bases induces a Riemannian metric $\g$, known as a \emph{reference} or an \emph{intrinsic} metric. 
The intrinsic metric is the metric with respect to which the bases are orthonormal; although no specific constitutive response is assumed ab initio, it is interpreted as the metric that a small neighborhood would assume if it were cut off from the rest of the body, and allowed to relax its elastic energy. 

The reference metric $\g$ induces also a Riemannian (Levi-Civita) connection, denoted $\nablaLC$, which differs from $\nabla$, unless the torsion vanishes. The Riemannian connection, unlike $\nabla$, is generally non-flat; its curvature, if non-zero, is an obstruction for the existence of a  strain-free global reference configuration.
Finally, a triple $(\M,\g,\nabla)$, where $\nabla$ is a flat connection, metrically consistent with $\g$, is known as a {\em Weitzenb\"ock space} or a {\em Weitzenb\"ock manifold} \cite{Wei23} (a notion originating from relativity theory, see e.g.~\cite{HS79,AP04}; for its use in the context of distributed dislocations, see e.g.~\cite{YG12b, OY14,KM15,KM15b}).

\paragraph{Constitutive paradigm}
The second paradigm, due largely to Noll \cite{Nol58} and Wang \cite{Wan67}, takes for elemental object a \emph{constitutive relation}. 
The underlying manifold $\M$ has for role to set the topology of the body, and be a domain for the constitutive relation. 
In the case of a hyperelastic body, the constitutive relation takes the form of an \emph{energy density} $W:T^*\M\otimes\R^d\to\R$.
A constitutive relation is called \emph{uniform} if the energy density at every point $p\in\M$ is determined by an ``archetypical" function $\calW: \R^d\times\R^d\to\R$, along with a local frame field $E: \M\times\R^d\to T\M$, which specifies how $\calW$ is ``implanted" into $\M$. 
Once a uniform constitutive relation has been defined, 
its pointwise symmetries and its dependence on position may define a so-called \emph{material connection} $\nabla$ along with an \emph{intrinsic Riemannian metric} $\g$ (described in detail in Section~\ref{sec:NW}).

At this point, it is interesting to note Wang's own reflections comparing the geometric approach (in our language) to his \cite{Wan67}:
\begin{quotation}
It is not possible to make any precise comparison, however, since the physical literature on dislocation theory rarely if ever introduces definite constitutive equations, resting content with heuristic discussions of the body manifolds and seldom taking up the response of bodies to deformation and loading, which is the foundation stone of modern continuum mechanics.
\end{quotation}
Indeed, in the geometric paradigm, the constitutive relation typically does not appear explicitly.
However, the geometric and the constitutive paradigms are consistent with each other.
On the one hand, {as shown by Wang, a constitutive relation subject to a uniformity property defines an intrinsic metric and a material connection (as will be shown below, the material connection is unique only if $\calW$ has a discrete symmetry group). On the other hand, a body manifold endowed with a notion of distant parallelism defines a uniform constitutive relation for every choice of archetypal function $\calW$ and implant map at a single point---once $\calW$ has been implanted at some $p\in \M$, the whole constitutive relation is determined by parallel transporting this implant to any other point in $\M$ according to $\nabla$;  by construction, $\nabla$ is a material connection of that constitutive relation.
 
This is the viewpoint that we take in this paper, and the one through which we show how homogenization processes in both paradigms are also equivalent with each other (see below).
However, Wang's comment above is not unfounded: first, in the case of an archetype with a continuous symmetry group (say, isotropic), there is more than one material connection associated with the constitutive relation, hence from the constitutive point of view it does not make sense to talk about a single parallelism (or Weitzenb\"ock manifold) that represents the body.
Second, in certain cases in which the geometric viewpoint assumes a posteriori a constitutive response, the parallelism, or the torsion tensor associated with it, are eventually considered as variables in the constitutive relation \cite{Kro96}, resulting in so-called {\em coupled stresses} \cite{Kro63b}.
This approach, in which the underlying geometric structure can change, e.g., due to loading, is beyond the scope of the constitutive paradigm (or at least, its time-independent version), and such models will not be considered in this work.

Finally, let us note that there are other approaches to dislocations not covered by the above discussion, which are beyond the scope of this paper.
In particular, we will not consider the line of works emanating from Davini \cite{Dav86}, and other more recent approaches such as \cite{Kat05,CK13}, although some of the consequences of the discussions here (e.g., continuous vs.~discrete symmetries) may also apply to them.

%%%%%%%%
\subsection{Description of the main results}
The physical notion of dislocations is rooted in discrete structures, such as defective crystal lattices. Thus, when considering distributed dislocations, it is natural to consider a homogenization process, in which a continuous distribution of dislocations (according to a chosen paradigm) is obtained as a limit of finitely many dislocations, as those are getting denser in some appropriate sense.
A priori, each of the two paradigms could have its own homogenization theory: 
\begin{enumerate}
\item Geometric paradigm: Consider body manifolds representing solids with finitely-many (singular) dislocations, and study their  limit as the number of dislocations tends to infinity.
\item Constitutive paradigm: Consider constitutive relations modeling solids with finitely-many (singular) dislocations, and study their  limit as the number of dislocations tends to infinity.
\end{enumerate}
The first task belongs to the realm of geometric analysis, and has been addressed in \cite{KM15,KM15b}, where it was shown that any two-dimensional Weitzenb\"ock manifold can be obtained as a limit of bodies with finitely-many dislocations (see Section~\ref{sec:homogenization_geo} for a precise statement).
The second task belongs, for hyperelastic bodies, to the realm of the calculus of variations, and has been addressed in \cite{KM16} for the special case of isotropic materials.

In this paper, we review the main results of these papers and extend the analysis of \cite{KM16} to the non-isotropic case. More importantly,
we show that the homogenization theories resulting from the geometric and the constitutive paradigms are consistent, and even identical in the case of constitutive relations having discrete symmetries. In particular, both predict the emergence of (the same) torsion as a limit of distributed dislocations. 

Our main result in this chapter can be summarized as follows:
\begin{theorem}[Equivalence of homogenization processes, informal]\
\label{thm:loop}
\begin{enumerate}
\item For a body manifold $(\M,\g,\nabla)$ with finitely many dislocations, there is a natural way to define 
a constitutive relation $(\M,W)$ based on a given archetype $\calW$, for which $\nabla$ is a material connection and $\g$ is an intrinsic metric (Proposition~\ref{pn:associated_energy}).
\item
If the archetype $\calW$ has a discrete symmetry group, then this relation is bijective; i.e., a constitutive relation $(\M,W)$ defines a unique material connection $\nabla$ and a unique intrinsic metric $\g$ (Proposition~\ref{pn:discrete_symmetry_equivalence}).
\item If a sequence of body manifolds with $n$ dislocations $(\M_n,\g_n,\nablaLC_n)$ converges (in the sense
of Theorem~\ref{thm:manifold_conv}) to a Weitzenb\"ock manifold $(\M,\g,\nabla)$, then the corresponding constitutive models $(\M,W_n)$ $\Gamma$-converge to a constitutive model $(\M,W)$, for which $\nabla$ is a material connection and $\g$ is an intrinsic metric (Theorem~\ref{thm:Gamma_conv}).

\end{enumerate}
\end{theorem}
\begin{figure}
\begin{center}
\hfil
\begin{xy}
%vertices
(60,0)*+{(\M,\g,\nabla)} = "M";%
(60,30)*+{W} = "W";%
(0,0)*+{(\M_n,\g_n,\nablaLC_n)} = "Mn";%
(0,30)*+{W_n} = "Wn";%
%arrows
{\ar@{->}^{\text{Prop.~\ref{pn:associated_energy}}} "M"; "W"};%
{\ar@{->}_{\text{Prop.~\ref{pn:associated_energy}}} "Mn"; "Wn"};%
{\ar@{->}_{\text{Thm.~\ref{thm:manifold_conv}}} "Mn"; "M"};%
{\ar@{->}^{\Gamma-\text{limit}\,\, \text{(Thm.~\ref{thm:Gamma_conv})}} "Wn"; "W"};%
{\ar@{-->}@/_{2pc}/_{\text{discr. symm.}}"Wn";"Mn"};
{\ar@{-->}@/^{2pc}/^{\text{discr. symm.}}"W";"M"};

\end{xy}
\hfil
\end{center}
\caption{A sketch of the main result (Theorem~\ref{thm:loop}).}
\label{fig:loop}
\end{figure}
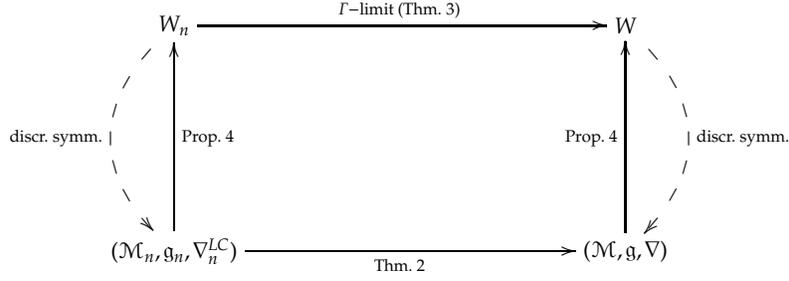
A sketch of Theorem~\ref{thm:loop} is shown in Figure~\ref{fig:loop}.

\medskip
In addition to Theorem~\ref{thm:loop}, this paper reviews the fundamental notions of the geometric and constitutive paradigms, and their above-mentioned equivalence;
we believe that the current presentation is original, and includes several results for which it is difficult (if at all possible) to find in the existing literature precise statements and proofs.

In the rest of this section, we elaborate on our main results.
We start by considering a defect-free body: in the geometric paradigm, such a body is modeled as a $d$-dimensional Riemannian manifold $(\M_0,\g_0)$, which can be embedded isometrically in Euclidean space $(\R^d,\euc)$, where $\euc$ is the standard Euclidean metric. 
Let $\nablaLC_0$ be the Levi-Civita connection of $\g_0$;  since $(\M_0,\g_0)$ is isometric to a Euclidean domain, the connection $\nablaLC_0$ is flat, and parallel transport is path-independent. 

To obtain a constitutive relation for that same body, one has to fix an archetype $\calW$ and a bijective linear map $(E_0)_p:\R^d\to T_p\M_0$ at some reference point $p\in\M_0$. The two together determine the mechanical response to deformation at $p$: for $A\in T^*_p\M_0\otimes\R^d$, the elastic energy density (per unit volume, where the reference volume is the volume form of $(\M_0,\g_0)$) at $p$ is
\[
(W_0)_p(A) = \calW(A\circ (E_0)_p).
\]
A constitutive relation is obtained by extending $(E_0)_p$ into a $\nablaLC_0$-parallel frame field $E_0:\M_0\times\R^d\to T\M_0$ (here is where the path-independence of the parallel transport is required). The elastic energy density  is
\beq
W_0(A) = \calW(A\circ E_0),
\label{eq:W_no_defects}
\eeq
and the elastic energy associated with a map $f:\M_0\to\R^d$ is 
\beq
I_0(f) = \int_{\M_0} W_0(df)\,\dVol_{\g_0},
\label{eq:energyI}
\eeq
where $\dVol_{\g_0}$ is the Riemannian volume form. As we show in Section~\ref{sec:equivalence}, the geometric and the constitutive paradigms are consistent: $\g_0$ is an \emph{intrinsic metric} for $W_0$ and $\nablaLC_0$ is a \emph{material connection} for $W_0$; 
moreover, $\nablaLC_0$ is the unique material connection for $W_0$, provided that $\calW$ has a discrete symmetry group.

Consider next a body with a single straight edge-dislocation. Since straight edge dislocations are in essence two-dimensional.  From the point of view of the geometric paradigm, the body manifold of a body with one edge-dislocation can be described by a Volterra cut-and-weld protocol \cite{Vol07}. There are numerous ways of implementing a Volterra protocol: two ways are depicted in Figure~\ref{fig:dislocation}.

\begin{figure}
\begin{center}
\begin{tikzpicture}[scale=1.2]
	%\draw[black,fill=gray!10] (0,0) circle (2cm);
	\draw[black,fill=gray!10] (-2,-2) -- (2,-2) -- (2,2) -- (-2,2) -- cycle;
	\fill[white] ((2.1,0.5) -- (0.5,0.5) -- (0,0) -- (0.5,-0.5) -- (2.1,-0.5) -- cycle;
	\draw[black] ((1.94,0.5) -- (0.5,0.5) -- (0,0) -- (0.5,-0.5) -- (1.94,-0.5) ;
	\tkzDefPoint(1.94,-0.5){Cp};
	\tkzDefPoint(0.5,-0.5){Ep};
	\tkzDefPoint(0,0){F};
	\tkzDefPoint(0.5,0.5){E};
	\tkzDefPoint(1.94,0.5){C};
	\tkzDrawPoints(Cp,Ep,F,E,C);
	\tkzLabelPoint[left](F){$p$}
	\tkzLabelPoint[right=2pt](C){$q$}
	\tkzLabelPoint[right=2pt](Cp){$q'$}
	\tkzLabelPoint[above](E){$r$}
	\tkzLabelPoint[below](Ep){$r'$}
	\tkzText(-1.4,0){$\M_1$}
	\tkzDefPoint(2,0){Fp};
	%\tkzDrawSegment[dashed](F,Fp);
	\tkzMarkAngle[size=0.5cm](Ep,F,E);
	\tkzText(0.3,0.0){$\theta$}
	\tkzText(0.15,-0.35){$d$}
\end{tikzpicture}

\begin{tikzpicture}[scale = 0.8]
	\fill[color=gray!5] (-3,-3.2) -- (5.5,-3.2) -- (5.5,3.2) -- (-3,3.2) -- cycle;
	\tkzDefPoint(0,0){O}
	\tkzDefPoint(-3,0.5){A}
	\tkzDefPoint(-3,-0.5){B}
	\tkzDrawSegment(O,A);
	\tkzDrawSegment(O,B)
	\tkzDrawArc[dotted](O,B)(A)
	\tkzMarkSegment[mark=s||](O,A)
	\tkzMarkSegment[mark=s||](O,B)
	\tkzDrawPoint(O)
	\tkzLabelPoint[above](O){$p$}
	\tkzMarkAngle[size=2.5cm](A,O,B);

	\tkzText(-1.9,0.0){$\theta$}

	\tkzDefPoint(1,0){pm}
	\tkzDrawPoint(pm)
	\tkzLabelPoint[above](pm){$r$}
	\tkzDefPoint(5,0){dummy}
	\tkzDrawSegment[dashed](pm,dummy)
	\tkzText(1.4,0){\ding{36}}
	\tkzText(2.2,0){\ding{36}}
	
	\tkzDefPoint(2.5,0){O1}
	\tkzDrawPoint(O1)

	\tkzDefPoint(5.5,0.5){A1}
	\tkzDefPoint(5.5,-0.5){B1}
	\tkzDrawSegment(O1,A1);
	\tkzDrawSegment(O1,B1)
	\tkzDrawArc[dotted](O1,B1)(A1)
	\draw[->] (4.0,0.75) -- (2.0,0.5);
	\draw[->] (4.0,-0.75) -- (2.0,-0.5);
	\tkzLabelPoint[above](O1){$r$}

\end{tikzpicture}

\end{center}
\caption{Two equivalent cut-and-weld constructions generating a body manifold with a single edge-dislocation. Top: the segments $pr$ and $pr'$ are identified (i.e., glued) as well as the segments $rq$ and $r'q'$. $p$ and $r\sim r'$ are the only singular points in the manifold (each with conical singularity of the same magnitude and opposite sign). 
Bottom: a sector whose vertex is denoted by $p$ is removed from the plane and its outer boundaries are glued together, thus forming a cone. The same sector is then inserted into a straight cut along a ray whose endpoint is denoted by $r$.}
\label{fig:dislocation}
\end{figure}
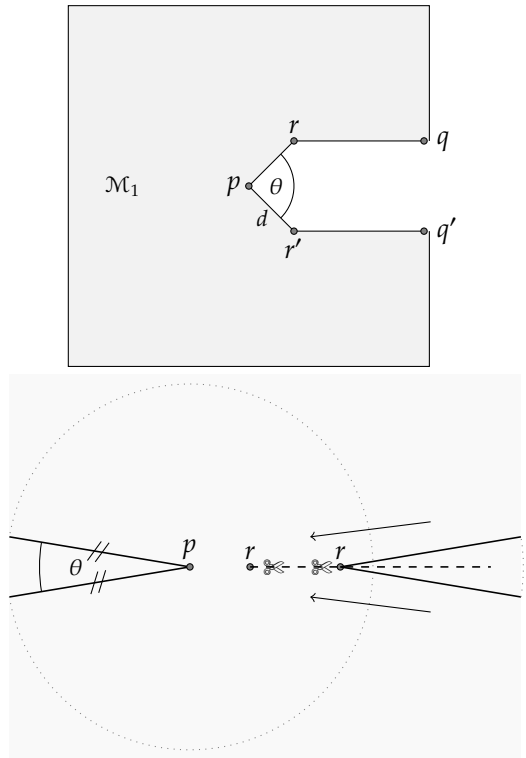

The outcome of this cut-and-weld protocol is a topological manifold $\M_1$, which is smooth everywhere except at two points (the points $p$ and $r\sim r'$ in Figure~\ref{fig:dislocation}). It is endowed with a metric $\g_1$, which is locally Euclidean,
since locally, every non-singular point has a neighborhood satisfying the above-mentioned defining properties of a defect-free body manifold. As there is no \emph{continuous} distribution of dislocations in this picture, the natural connection associated with this body is the Levi-Civita connection $\nablaLC_1$ of $(\M_1,\g_1)$.
Moreover, the parallel transport induced by $\nablaLC_1$ is path-independent for all closed paths that do not encircle only one of the two singular points. This restriction on admissible paths can be replaced by removing from the smooth part a segment connecting the two singular points. 
Note that the \emph{topological} manifold $\M_1$ is simply-connected, however its smooth component is not.  Despite being (almost everywhere) locally-Euclidean, it cannot be embedded in the Euclidean plane isometrically. 

The procedure for obtaining a constitutive relation within the constitutive paradigm follows the exact same lines as for a defect-free body. One has to fix an archetype $\calW$ and a frame at a point $(E_1)_p:\R^2\to T_p\M_1$; extending $(E_1)_p$ into a parallel frame field $E_1:\M_1\otimes\R^2\to T\M_1$, the elastic energy density $W_1: T^*\M_1\otimes\R^2\to\R$ is given by \eqref{eq:W_no_defects}, after changing the subscript $0$ to $1$. Once again, the two paradigms are consistent, as $\g_1$ and $\nablaLC_1$ are an intrinsic metric and a material connection for the energy density $W_1$. Note that none of the two pictures makes any explicit mention of torsion.

The generalization of this procedure to a body carrying $n$ singular edge dislocations follows the same lines, performing $n$ Volterra cut-and-weld protocols, thus obtaining a simply-connected topological manifold, which is smooth everywhere but at $n$ pairs of singular point. On the geometric side, one obtains a triple $(\M_n,\g_n,\nablaLC_n)$, where the Levi-Civita connection $\nablaLC_n$ has trivial holonomy, namely, its parallel transport is path-independent for all closed paths that do not encircle only one singular point within a pair. After the choice of an archetype $\calW$ and a frame at a point $(E_n)_p: \R^2\to T_p^*\M_n$, one obtains an energy density $W_n$, for which $\g_n$ and $\nablaLC_n$ are an intrinsic metric and a material connection.

Next consider the limit of $n\to\infty$. 
As proved in \cite{KM16}, every two-dimensional body manifold $(\M,\g,\nabla)$ admitting a global $\nabla$-parallel frame field is a limit of manifolds $(\M_n,\g_n\nablaLC_n)$ with finitely-many dislocations. A precise definition of this convergence is stated in Theorem~\ref{thm:manifold_conv}; loosely speaking, it means that $\M_n$ can be mapped into $\M$ such that orthonormal $\nablaLC_n$-parallel frame fields $E_n$ are mapped into a frame field asymptotically close to an orthonormal $\nabla$-parallel frame field. Note the {\em emergence of torsion}, as $\nablaLC_n$ is torsion-free for every $n$, whereas $\nabla$ has non-zero torsion. 

We then switch to the constitutive paradigm: as described above, each of the manifolds $(\M_n,\g_n\nablaLC_n)$ defines, upon the choice of an archetype $\calW$ and a frame at one point $(E_n)_p$, an energy density $W_n$,  and an associated energy $I_n$. 
In Theorem~\ref{thm:Gamma_conv}, we prove that as $n\to\infty$,  $I_n$ converges in the sense of $\Gamma$-convergence to a limiting functional $I$, which has an energy density $W$, where $W$ is the energy density obtained by the same construction using $\calW$ and $\nabla$.
In particular, $W$ has intrinsic metric $\g$ and material connection $\nabla$. 
This ``closes the circle", proving that the construction of a uniform energy density from a given body manifold can be extended from finitely-many to continuously-distributed dislocations.  

%%%%%%%%%%%%%%%%%%%%%%%%
\subsection{Structure of this paper}

In the rest of the paper we formalize the above outline:
\begin{itemize}
\item In Section~\ref{sec:NW}, we present the main ingredients of the constitutive paradigm, following \cite{Wan67}, under the assumption of hyperelasticity. 
We use a more modern notation and some simplifying assumptions.

Furthermore, we show in Section~\ref{sec:NW} how the constitutive paradigm and the geometric paradigm for describing dislocations are consistent, and equivalent in the case of discrete symmetry group (Propositions~\ref{pn:associated_energy}-\ref{pn:discrete_symmetry_equivalence}), thus establishing the vertical arrows in Figure~\ref{fig:loop}.

\item In Section~\ref{sec:homogenization_geo}, we present in more detail the modeling of dislocations via the geometric paradigm, using the notion of Weitzenb\"ock manifolds.
In particular, we explain how Burgers vectors arise in this context, and their relation to the torsion tensor.

The main part of this section is an overview of recent results \cite{KM15,KM15b} concerning the homogenization of dislocations within this paradigm---a convergence of Weitzenb\"ock manifolds (Theorem~\ref{thm:manifold_conv}).
This establishes the lower horizontal arrow in Figure~\ref{fig:loop}.
For the sake of readability, we omit some of the technical details, and focus on the main ideas of the construction.

\item In Section~\ref{sec:homogenization}, we prove the convergence of the elastic energies associated with the converging Weitzenb\"ock manifolds; we show that they $\Gamma$-converge to the elastic energy associated with the limiting Weitzenb\"ock manifold (Theorem~\ref{thm:Gamma_conv}, Corollary~\ref{cor:loop}), thus establishing the upper horizontal arrow in Figure~\ref{fig:loop}, and concluding the proof of Theorem~\ref{thm:loop}. 

\item Finally, in Section~\ref{sec:appendix} we show explicitly how the torsion tensor appears in the equilibrium equations of elastic bodies with continuously-distributed dislocations according to the constitutive paradigm.
\end{itemize}

%%%%%%%%%%%%%%%%%%%%%%%%%%%%%%%%%%%%%%%%%%%%%
\section{The constitutive paradigm of Noll and Wang}
\label{sec:NW}

In this section we present some of the basic notions of the Noll-Wang approach. 
We generally follow \cite{Wan67}, although our presentation and some of the proofs are somewhat different.
For simplicity, we will assume a hyperelastic model.

%%%%%%%%%
\begin{definition}[Hyperelastic body]
\label{def:elastic_body_NW}
A hyperelastic body consists of a $d$-dimensional differentiable manifold, $\M$---the \emph{body manifold}---and an energy-density function (or \emph{constitutive relation}),
\[
W: T^*\M\otimes\R^d \to \R,
\]
which is viewed as a (nonlinear) bundle map over $\M$.
\end{definition}
%%%%%%%%

For $p\in\M$ and $A\in T^*_p\M\otimes\R^d$, we denote the action of $W$ on $A$ by $W_p(A)$.
If $\xi$ is a section of $T^*\M\otimes\R^d$, then $W(\xi)$ is a function on $\M$.

\begin{remark}
In the terminology of Noll, such a body is called a \emph{simple body} since the constitutive relation at a point depends only on the local deformation (i.e., the first jet of the deformation) at that point.
\end{remark}

We will use the following notation: the groups $\GL(d)$, $\SO(d)$ are the standard subgroups of $\Hom(\R^d,\R^d)$; 
for two oriented inner-product spaces $(V,\g)$, $(W,\h)$ we will denote by $\SO(V,W)$ or $\SO(\g,\h)$ the set of orientation-preserving isometries $V\to W$, 
and by $\SO(V)$ the orientation-preserving isometries $V\to V$.

The next definition makes precise the notion of material uniformity, namely, a constitutive relation that is ``the same" at every point:

%%%%%%%%%
\begin{definition}[Material uniformity]
\label{def:uniformity_NW}
A hyperelastic body is called \emph{uniform} if for every $p\in \M$ there exists a frame, i.e., a linear isomorphism $E_p:\R^d\to T_p\M$ such that,
\beq
\label{eq:implant}
W_p(A) = \calW(A\circ E_p) \quad \text{for every } A\in T_p^*\M\otimes \R^d,
\eeq
for some 
\[
\calW : \R^d\otimes \R^d \to \R
\]
independent of $p$.
\end{definition}
%%%%%%%%

\begin{remark}
More precisely, a hyperelastic energy density $W$ is a section of $(T^*\M\otimes\R^d)^*\otimes \wedge^dT^*\M$, i.e., for $A\in  T^*\M\otimes\R^d$, $W(A)$ is a $d$-form. Correspondingly, a body is uniform if there exists an archetype 
\[
\calW : \R^d\otimes \R^d\to\wedge^d\R^d
\]
such that $W_p = (E_p)^*\calW$. 
Since, eventually, we will only consider solid bodies with a given Riemannian volume form, it is more convenient to consider $W$ as a scalar density with respect to this volume form, and $\calW$ is a scalar density with respect to the canonical volume form in $\R^d$.
The given volume form then appears when considering the energy functional and not merely the scalar energy density, as in \eqref{eq:energyI} or Definition~\ref{def:I_tildeI}.
\end{remark}

Material uniformity is the weakest sense in which a constitutive relation is independent of position; it is defined independently of any coordinate system.  It is a type of what is sometimes called ``homogeneity" (though this term has another significance in \cite{Wan67}). The function $\calW$ is sometimes called an \emph{archetype}, whereas
the frame $E_p$ is sometimes called an \emph{implant map}, because it shows how the archetype $\calW$ is implanted into the material.
Note that for a given uniform constitutive relation, neither the archetype $\calW$ nor the implant map $E_p$ are unique.
If $(\calW,E_p)$ is an archetype-implant pair at $p\in\M$, then so is $(\calW',E_p\circ S)$, where $S\in\GL(d)$, and for every $B\in\Hom(\R^d,\R^d)$,
\[
\calW'(B) = \calW(B\circ S^{-1}).
\]
Moreover, the implant map may not be unique even for a fixed $\calW$, depending on the symmetries of $\calW$ (see below).

%%%%%%%%%
\begin{definition}[Smooth body]
A uniform hyperelastic body is called \emph{smooth} if there exists an archetype $\calW$, a cover of $\M$ with open sets $U^\alpha$, and implants $E^\alpha = \{E_p^\alpha\}_{p\in U^\alpha}$, such that the sections $E^\alpha$ are smooth. 
\end{definition}
%%%%%%%%

%%%%%%%%%%%%%%%
\begin{example}
\label{ex:dist2}
Let $\g$ be a smooth Riemannian metric on $\M$, and consider the energy density
\beq
W(A) = \dist^2(A,\SO(\g,\euc)),
\label{eq:dist2}
\eeq
where $\SO(\g,\euc)$ at $p\in\M$ is the set of orientation-preserving isometries $T_p\M\to\R^d$, and the distance in $T^*_p\M\otimes\R^d$ is induced by the inner-product $\g_p$ on $T\M$ and the Euclidean inner-product $\euc$ on $\R^d$.
Then, any orthonormal frame $E_p\in \SO(\R^d,T_p\M)$ is an implant map, with archetype
\beq
\label{eq:dist2archtype}
\calW(\cdot) = \dist^2(\cdot, \SO(d)).
\eeq
This body is smooth, as we can choose locally smooth orthonormal frames.
Note that the implant map is non-unique, as it may be composed with any smooth section of $\SO(d)$ over $\M$.
This example illustrates why we do not require the existence of a global section $\{E_p\}_{p\in \M}$ in the definition of smoothness; such sections may not exist regardless of $W$, for example because of topological obstructions on $\M$ (e.g., if $\M$ is a sphere).
\end{example}
%%%%%%%%%%%%%%%

%%%%%%%%%
\begin{definition}[symmetry group]
\label{def:isotropy}
Let $\M$ be a uniform hyperelastic body. 
The \emph{symmetry group} of the body associated with an archetype $\calW$ is a group $\mathcal{G} \le \GL(d)$, defined by
\[
\calW(B\circ g) = \calW(B) \quad \text{for every $B\in \R^d\otimes \R^d$ and $g\in \mathcal{G}$}.
\]
The body is called a \emph{solid} if there exists a $\calW$ such that $\mathcal{G} \le \SO(d)$ (or sometimes if $\mathcal{G}\le O(d)$).
In this case, we shall only consider such $\calW$ as admissible, and call $\calW$ \emph{undistorted}.
\end{definition}
%%%%%%%%

It is easy to see that if $\calW$ and $\calW'$ are archetypes for the same constitutive relation, then their symmetry groups $\mathcal{G}$ and $\mathcal{G}'$ are conjugate, i.e., there exists a $g\in\GL(d)$, such that $\mathcal{G}' = g^{-1}\mathcal{G} g$. Thus, a hyperelastic body is a solid if and only if it has an archetype $\calW$, whose symmetry group is conjugate to a subgroup of $\SO(d)$.

The intrinsic right-symmetry of the constitutive relation is determined by $W$ rather than by $\calW$. The symmetry group of $W$ is a point $p\in\M$ is a subgroup
\[
\mathcal{G}_p \le \GL(T_p\M).
\]
If $(\calW,E_p)$ is an archetype-implant pair at $p$, and $\mathcal{G}$ is the symmetry group of $\calW$, then for every $g\in\mathcal{G}$ and $A\in T^*_p\M\otimes\R^d$,
\[
W_p(A) = \calW(A\circ E_p) = \calW(A\circ E_pg) = W_p(A\circ E_pgE_p^{-1}),
\]
i.e.,
\[
\mathcal{G}_p = E_p \mathcal{G} E_p^{-1}.
\]
Consequently, the space of all implant maps that correspond to $\calW$ at $p$ is $E_p \mathcal{G}$.

\begin{example}
In Example~\ref{ex:dist2}, the symmetry group of $W$ at $p\in\M$ is $\SO(T_p\M)$ (where $T_p\M$ is endowed with the inner-product $\g_p$). 
$\calW$ is undistorted if and only the implant map $E_p$ at every $p\in\M$ satisfies 
\[
E_p^{-1} \SO(T_p\M) E_p = \SO(d).
\]
In particular, the archetype \eqref{eq:dist2archtype} is undistorted.
\end{example}

Thus far, we only considered point-symmetries of $W$ in the form of symmetry groups. We next consider symmetries of $W$ associated with pairs of points in the manifold:  

%%%%%%%%%
\begin{definition}[Material connection]
\label{def:material_connection}
A \emph{material connection} of $(\M,W)$ is an affine connection $\nabla$ on $\M$ whose parallel transport operator $\Pi$ leaves $W$ invariant.
That is, for every $p,q\in \M$, $A\in T^*_q\M\otimes \R^d$ and path $\gamma$ from $p$ to $q$,
\[
W_p(A\circ \Pi_\gamma) = W_q(A),
\]
where $\Pi_\gamma:T_p\M\to T_q\M$ is the parallel transport along $\gamma$.
\end{definition}
%%%%%%%%

In general, a material connection may fail to exist (there may be topological obstructions), or may not be unique. The following proposition relates the uniqueness of a material connection to the nature of the symmetry group (a less general version of this result appears in \cite{Wan67}):

%%%%%%%%%%
\begin{proposition}
\label{pn:discrete_symmetry}
Let $(\M,W)$ be a smooth uniform hyperelastic body with symmetry group $\mathcal{G}$.
If $\mathcal{G}$ is discrete, then there exists a unique locally-flat material connection.\footnote{Strictly speaking, the intrinsic condition is that $\mathcal{G}_p$ is discrete for some $p\in\M$ (and therefore for every $p\in\M$). By locally-flat, we mean that the curvature tensor vanishes; globally-flat implies also a trivial holonomy. Note that the term \emph{flat} has a different interpretation in \cite{Wan67}, where it describes a curvature- and torsion-free connection.}
\end{proposition}
%%%%%%%%%

%%%%%%%
\begin{proof}\smartqed
Assume two material connections, whose parallel transport operators are $\Pi^1$ and $\Pi^2$. Let $\gamma$ be a curve starting at $p\in \M$, and let $A\in T^*_{\gamma(t)}\M\otimes \R^d$ for some $t\ge0$.
Then,
\[
W_p(A\circ \Pi^1_{\gamma|_{[0,t]}} ) = W_{\gamma(t)}(A) = W_p(A\circ \Pi^2_{\gamma|_{[0,t]}}).
\]
Setting $A = B\circ (\Pi^1_{\gamma|_{[0,t]}})^{-1}$ for $B\in T^*_p\M\times\R^d$,
we obtain that
\[
W_p(B)  = W_p(B\circ (\Pi^1_{\gamma|_{[0,t]}})^{-1}\circ \Pi^2_{\gamma|_{[0,t]}}),
\]
hence
\[
(\Pi^1_{\gamma|_{[0,t]}})^{-1} \Pi^2_{\gamma|_{[0,t]}} \in \mathcal{G}_p
\]
for every $t$. Since the left-hand side is continuous in $t$ and $\mathcal{G}_p$ is a discrete group, $(\Pi^1_{\gamma|_{[0,t]}})^{-1} \Pi^2_{\gamma|_{[0,t]}}$ is constant. 
Since at $t=0$ it is the identity, 
\[
\Pi^1_{\gamma|_{[0,t]}} = \Pi^2_{\gamma|_{[0,t]}}
\]
for every $t$. Finally, since $\gamma$ is arbitrary,  $\Pi^1 = \Pi^2$.

We next prove existence of a locally-flat material connection.
Let  $\cup_\alpha U^\alpha = \M$ be a cover of $\M$, and let $\{E^\alpha_p\}_{p\in U^\alpha}$ be implant maps. 
For a curve $\gamma\subset U^\alpha$ starting at $p$ and ending at $q$, define
\beq
\label{eq:implants_define_connection}
\Pi_\gamma = E^\alpha_{q} \circ (E_{p}^\alpha)^{-1}.
\eeq
For a general curve $\gamma\subset\M$, partition it into curves $\gamma = \gamma_n * \ldots *\gamma_1$ (where $*$ is the concatenation operator), where each $\gamma_i \subset U^{\alpha_i}$ for some $\alpha_i$, and use the above definition.
In order to show that $\Pi_\gamma$ is well-defined, we need to show that this definition is independent of the concatenation.
To this end, it is enough to show that if $\gamma\subset U^\alpha \cap U^\beta$, then the definition of $\Pi_\gamma$ with respect to either $U^\alpha$ or $U^\beta$ is the same.

Indeed, consider the function of $p$,
\[
(E^\alpha_p)^{-1}E^\beta_p : U^\alpha\cap U^\beta \to \GL(R^d).
\] 
Since for any $A\in T^*_p\M\otimes \R^d$,
\[
\calW(A\circ E^\alpha_p)  = W(A) = \calW(A\circ E^\beta_p),
\]
it follows that $(E^\alpha_p)^{-1}E^\beta_p \in \mathcal{G}$ for any $p\in U^\alpha\cap U^\beta$.
Since $\mathcal{G}$ is discrete, it follows that this is a constant function of $p$, that is $(E^\alpha_p)^{-1}E^\beta_p = B\in \GL(\R^d)$ for every $p$.
We therefore have that for $p,q\in U^\alpha\cap U^\beta$,
\[
E^\alpha_{q}  (E_{p}^\alpha)^{-1} = E^\alpha_{q}  B  B^{-1} (E_{p}^\alpha)^{-1} = E^\alpha_{q} (E^\alpha_q)^{-1}E^\beta_q (E^\beta_p)^{-1}E^\alpha_p (E_{p}^\alpha)^{-1} = E^\beta_q (E^\beta_p)^{-1},
\]
and therefore $\Pi_\gamma$ is well defined.
Finally, for a closed curve $\gamma$, starting and ending at $p$, and contained in one of the domains $U^\alpha$, it follows from the definition that $\Pi_\gamma = \textup{Id}_{T_p\M}$, hence the holonomy of $\Pi$ is locally trivial, which implies that the curvature tensor of the connection associated with $\Pi$ is zero.
Note, however, that the holonomy of $\Pi$ may be non-trivial in general (for non-simply-connected manifolds).
\qed\end{proof}
%%%%%%

Note that if there exists a global continuous implant section $\{E_p\}_{p\in \M}$ (for an archetype $\calW$), then the connection defined by \eqref{eq:implants_define_connection} (without the $\alpha$ superscript) is well-defined regardless of the symmetry group, and moreover, it is not only locally-flat, but has a trivial holonomy (that is, a path-independent parallel transport).
In fact, the existence of a material connection with a trivial holonomy is equivalent to the existence of a global implant section $\{E_p\}_{p\in \M}$.
Indeed, let $\nabla$ be such a connection, and let $E_{p_0}$ be an implant at $p_0\in \M$. then
\beq
\label{eq:connection_defines_implant}
E_p := \Pi_\gamma E_{p_0}
\eeq
is a global continuous implant section (here $\gamma$ is an arbitrary curve connecting $p_0$ and $p$).

In the case of a solid body, there is an additional intrinsic geometric construct associated with the body:

%%%%%%%%%
\begin{definition}[Intrinsic metric]
\label{def:intrinsic_metric}
Let $(\M,W)$ be a smooth solid body with an undistorted archetype $\calW$ and implant maps $\{E_p\}_{p\in\M}$.
The \emph{intrinsic Riemannian metric} of $\M$ associated with $\calW$ is defined by
\beq
\label{eq:implants_define_metric}
\g_p(X,Y) = \euc(E_p^{-1}(X), E_p^{-1}(Y)), \quad \text{ for every } X,Y\in T_p\M,
\eeq
where $\euc$ is the Euclidean inner-product in $\R^d$.
\end{definition}
%%%%%%%%

This definition depends on $\calW$ (see Example~\ref{ex:intrinsic_metrics} below), but not on the choice of implants $E_p$.
Indeed, if $E_p$ and $E_p'$ are two implants at $p$, then, since $\M$ is a solid, $g=E_p^{-1}E_p'  \in \mathcal{G} \le \SO(d)$, and therefore
\[
\euc(E_p^{-1}(X), E_p^{-1}(Y)) = \euc\brk{g\circ {E_p'}^{-1} (X), g\circ {E_p'}^{-1} (Y)} = \euc({E_p'}^{-1} (X), {E_p'}^{-1}(Y)),
\]
where we used in the last step the $\SO(d)$ invariance of the Euclidean metric.
Note also that the existence of a Riemannian metric on $\M$ that is invariant under the action of $\mathcal{G}_p$ implies that $\M$ is solid \cite[Proposition~11.2]{Wan67}.

%%%%%%%%%%%%%%%
\begin{proposition}
\label{pn:metric_connection}
If $\nabla$ is a material connection and $\g$ is an intrinsic metric of a solid $\M$ with an archetype $\calW$, then $\nabla$ is metrically-consistent with $\g$ (equivalently, the induced parallel transport is an isometry).
\end{proposition}

\begin{proof}\smartqed
Let $p,q\in \M$, and let $\gamma$ be a curve from $p$ to $q$.
Let $\Pi_\gamma$ the parallel transport of $\nabla$ along $\gamma$,  $X,Y\in T_p\M$, and let $E_q$ be an implant at $q$.
Then
\[
\g_q(\Pi_\gamma X,\Pi_\gamma Y) = \euc(E_q^{-1} \circ \Pi_\gamma X, E_q^{-1}\circ \Pi_\gamma Y) = \g_p(X,Y),
\]
where in the right-most equality we used the fact that $\Pi_\gamma^{-1}\circ E_q$ is an implant at $p$, for the same archetype $\calW$.
This equality shows that $\nabla$ is metrically consistent with $\g$.
\qed\end{proof} 

\begin{corollary}\cite[Proposition~11.6]{Wan67}
A solid body $(\M,W)$ is equipped with at most one torsion-free material connection, in which case it is the Levi-Civita connection of all intrinsic metrics of $(\M,W)$.
\end{corollary}

%%%%%%%%%%%%%%%%
Proposition~\ref{pn:metric_connection} states that all material connections are metrically-consistent with every intrinsic metric.
In isotropic solids, i.e., solids whose symmetry group is $\SO(d)$, the converse is also true: every metrically-consistent connection is a material connection (note the strong contrast to the case of a discrete symmetry group, Proposition~\ref{pn:discrete_symmetry}):

\begin{proposition}
Let $(\M,W)$ be an isotropic solid and let $\nabla$ be a connection metrically-consistent with some intrinsic metric $\g$.
Then $\nabla$ is a material connection.
In particular, any isotropic solid admits a torsion-free connection---the Levi-Civita connection of any intrinsic metric.\footnote{This proposition is a more general version of \cite[Proposition~11.8]{Wan67}.}
\end{proposition}

\begin{proof}\smartqed
Let $\calW$ be an undistorted archetype and  let $E=\{E_p\}_{p\in \M}$ be an implant map (the proof below does not require any smoothness assumptions of $E$, and thus we can assume the existence of a global implant map without loss of generality). Suppose that $\g$ is an intrinsic metric for $W$, and let $\nabla$ be an affine connection metrically-consistent with $\g$; 
Since $(\M,W)$ is isotropic and $\calW$ is undistorted, we have (by definition) that its symmetry group is $\SO(d)$. 

Let now $\Pi_\gamma:T_p\M\to T_q\M$ be the parallel transport of $\nabla$ along a curve $\gamma$ from $p$ to $q$.
Since $\nabla$ is metrically-consistent with respect to $\g$, $\Pi_\gamma\in \SO(\g_p,\g_q)$.
Using the fact that for any $r\in\M$, $E_r\in \SO(\euc,\g_r)$ (by the very definition of an intrinsic metric), we have that $E_q^{-1} \circ \Pi_\gamma \circ E_p \in \SO(d)$.
Therefore, since $\calW$ is $\SO(d)$-invariant, we have that for any $A\in T^*_q\M\otimes \R^d$,
\[
W_p(A\circ \Pi_\gamma) = \calW(A\circ \Pi_\gamma \circ E_p) = \calW(A\circ E_q \circ (E_q^{-1}\circ \Pi_\gamma \circ E_p)) = \calW(A\circ E_q) = W_q(A).
\]
\qed\end{proof}

The fact that an isotropic solid always has a torsion-free material connection (or more generally, it has many material connections with different torsions) suggests that the equilibrium equations of such a body are independent of the torsion tensor.
Indeed, it can be shown explicitly (see Section~\ref{sec:appendix}) that $W$ only depends on the metric.

%%%%%%%%%%%%%%%
\begin{example}
\label{ex:intrinsic_metrics}
Consider once again Example~\ref{ex:dist2}.
Then $\g$ is an intrinsic metric, corresponding to the archetype
\[
\calW(B) = \dist^2(B, \SO(d)),
\]
and implants $E_p\in\SO(\euc,\g)$.
However, $c^2\g$, $c>0$, is also an intrinsic metric, corresponding to the archetype
\[
\calW(B) = \dist^2(cB, \SO(d))
\]
and implants $E_p\in c^{-1}\SO(\euc,\g)$.
It can be shown that there are no other intrinsic metrics in this case.
The phenomenon whereby the intrinsic metric is unique up to a multiplicative constant holds for every isotropic solid. 
\end{example}
%%%%%%%%%%%%%%% 

\begin{remark}
In two dimensions, a solid archetype is either isotropic or it has a discrete symmetry; 
in three dimensions, a body can also be transversely-isotropic (see \cite[p.~60]{Wan67}).
In this case, the material connection is not unique, but the Levi-Civita connection of an intrinsic metric may not be a material connection.
More on transversely-isotropic materials can be found in \cite[Proposition~11.9]{Wan67} and \cite[Proposition~5]{EES90}.
\end{remark}

%%%%%%%%%%%%%%%%%%%%%%%%%%%%%%%%%%%%%%%%%%%%%%%%%%%%%%%
\subsection{Relation between geometric and constitutive paradigms}
\label{sec:equivalence}

As presented in the introduction,  a body with distributed dislocations is modeled in the geometric paradigm as a Weitzenb\"ock manifold $(\M,\g,\nabla)$, where $\nabla$ is curvature-free and metrically consistent with $\g$.
For simplicity, assume that $\nabla$ also has trivial holonomy (an assumption that often appears implicitly in this paradigm), hence the parallel-transport operator of $\nabla$ is path-independent (a property known as \emph{distant parallelism} or \emph{teleparalellism}).
We denote the parallel transport from $p$ to $q$ by $\Pi_p^q$.

To relate the geometric body manifold to the constitutive paradigm, assume a given undistorted solid archetype $\calW$ and an implant $E_{p_0}$, which is an orthonormal basis (with respect to $\g_{p_0}$) at some $p_0\in \M$. 
The pair $(\calW,E_{p_0})$ determines the mechanical response of the body at the point $p_0$.
Parallel transporting $E_{p_0}$ using \eqref{eq:connection_defines_implant}, we obtain a parallel frame field $\{E_p\}_{p\in \M}$,
which is orthonormal, since $\nabla$ is metrically-consistent with $\g$.

An implant field $E=\{E_p\}_{p\in \M}$ and an archetype $\calW$ define a unique energy density using \eqref{eq:implant}.
Note that this is the only energy density $W$ with a material connection $\nabla$ for which $\calW$ is an archetype with an implant $E_{p_0}$ at $p_0\in \M$.

We have thus proved the following:

%%%%%%%%%%%%%%%%%%%%%%%%%
\begin{proposition}
\label{pn:associated_energy}
Fix a solid (undistorted) archetype $\calW\in C(\R^d\times\R^d)$.
\begin{enumerate}
\item Given a Weitzenb\"ock manifold $(\M,\g,\nabla)$ with trivial holonomy and an orthonormal basis $E_p\in\SO(\euc,\g_p)$ at some $p\in \M$, there exists a unique energy density $W$, such that $\M$ is uniform with archetype $\calW$, and implant map $E_p$ at $p$, and such that $\g$ is an intrinsic metric and $\nabla$ is a material connection.
\item Moreover, all energy densities $W$ having an archetype $\calW$, an intrinsic metric $\g$ and a material connection $\nabla$ can be constructed this way.
In particular, $W$ is unique up to a global rotation---the choice of a basis at one point.
\end{enumerate}
\end{proposition}

A somewhat more intrinsic version of this proposition would be that a Weitznb\"ock manifold $(\M,\g,\nabla)$ and a response function $W_p$ at a single point, defines a unique energy density consistent with $\g$ and $\nabla$ (without the need to define $E_p$ and $\calW$).
However, since the same archetype $\calW$ can be implanted into different bodies (thus making sense of different bodies having "the same" response function), and since we are eventually interested in this paper in sequences of elastic bodies, it is useful to take $\calW$ as a basic building block, as done in Proposition~\ref{pn:associated_energy}.

In the case of a discrete symmetry group, the constitutive model $(\M,W)$ induces a unique geometric model $(\M,\g,\nabla)$;
this follows readily from Proposition~\ref{pn:discrete_symmetry}, and the discussion following Definition~\ref{def:intrinsic_metric}:

\begin{proposition}
\label{pn:discrete_symmetry_equivalence}
Let $(\M,W)$ be a uniform solid material with an undistorted archetype $\calW$ having a discrete symmetry group.
Then, the material connection $\nabla$ and the intrinsic metric $\g$ associated with $\calW$ are unique
(if the symmetry group is not discrete, $\g$ is still uniquely determined, however not $\nabla$).
\end{proposition}

Another way of describing the relation between the geometric and constitutive paradigms is the following:
\begin{enumerate}[label=(\alph*)]
\item The triple $(\M,\calW,E)$, where $\calW:\R^d\otimes \R^d \to [0,\infty)$ and $E$ is a frame field, determines a uniform body $(\M,W)$ uniquely by \eqref{eq:implant}.
\item On the other hand, by declaring $E$ to be a parallel-orthonormal field, we obtain a Weitzenb\"ock manifold $(\M,\g,\nabla)$.
\end{enumerate}

In fact, $(\M,\calW,E)$ contains slightly more information than both $(\M,W)$ and $(\M,\g,\nabla)$:
given $\calW$, $(\M,\calW,E)$ can be derived from $(\M,\g,\nabla)$ uniquely, up to a global rotation (choice of $E_p$ at one point), and in the case of a discrete symmetry group, the same holds for deriving $(\M,\calW,E)$ from $(\M,W)$.

%%%%%%%%%%%%%%%%%%%%%%%%%%%%%%%%%%%%%%%%%%%%%%%%%%%%%%%
\section{Homogenization of dislocations: geometric paradigm}
\label{sec:homogenization_geo}

In this section we describe the results of \cite{KM15,KM15b}, showing how a smooth Weitzenb\"ock manifold $(\M,\g,\nabla)$, representing a body with continuously-distributed dislocations (the torsion tensor of $\nabla$ representing their density), can be obtained as a limit of bodies with finitely many dislocations.
These results are for two-dimensional bodies, hence we are only considering edge dislocations.

\paragraph{Bodies with finitely-many edge dislocations.}
To set the scene for the geometric homogenization of elastic bodies, we start by defining a two-dimensional body with finitely many (edge) dislocations.
As illustrated in Figure~\ref{fig:dislocation}, we view each dislocation as a pair of disclinations of opposite sign (a curvature dipole).

\begin{definition}
\label{def:body_w_dislocations}
A \emph{body with finitely-many singular edge dislocations} is a compact two-dimensional manifold with boundary $\M$, endowed with Riemannian metric $\g$, which is almost-everywhere smooth and locally-flat. The singularities are concentrated on a finite, even number of points, such that
\begin{enumerate}
\item The metric $\g$, restricted to a small enough neighborhood around a singular point, is a metric of a cone.
\item One can partition the singular points into pairs (\emph{curvature dipoles}), such that the geodesics connecting each pair (\emph{dislocation cores}) do not intersect.
\item The Levi-Civita connection $\nablaLC$, defined on the complement of those segments is path-independent.
\end{enumerate}
A body with finitely-many dislocations is a Weitzenb\"ock manifold $(\M,\g,\nablaLC)$,
and whenever we refer to a smooth field over $\M$ (say a frame field), it is understood as being smooth on complement of the dislocation cores. 
\end{definition}

The assumption on the Levi-Civita connection being path-independent, implies that the two cone defects in each pair (that is, the difference between $2\pi$ and the total angle around the cone) are of the same magnitude but of different signs.
That is, they are curvature dipoles.
In particular, the construction in Figure~\ref{fig:dislocation} yields a body with a single dislocation according to this definition.

Another approach for modeling bodies with finitely many dislocation was presented in \cite{ES14,ES14b}; 
instead of assuming a frame field describing lattice directions, one assumes a co-frame, that is, a family of $1$-forms (called {\em layering forms}).
This slightly different viewpoint enables the use of distributional $1$-forms---{\em de-Rham currents}---for describing the singular dislocations.
This viewpoint is quite close to the one presented here, although in some sense it requires less structure.
Recently, a homogenization result in this context has been proved \cite{KO19}, which is similar conceptually to the one presented here.
However, the notion of convergence used in \cite{KO19} is very weak compared to Theorem~\ref{thm:manifold_conv}, and therefore much more difficult to relate to the convergence of associated energy functionals, which is the main result of this paper.

%%%%%%%%%%%%%%%%%%%%%%%
\paragraph{Burgers circuits and vectors}
We now present in more detail how Burgers vectors appear in the context of Weissenb\"ock manifolds.
Let $\M$ be a manifold, endowed with a connection $\nabla$. 
A {\em Burgers circuit} is a closed curve $\gamma:[0,1]\to \M$, and its associated {\em Burgers vector} is defined by 
\[
\b_\gamma = \int_0^1 \Pi_{\gamma(t)}^{\gamma(0)} \dot{\gamma}(t) \,dt \in T_{\gamma(0)} \M,
\]
where $\Pi_{\gamma(t)}^{\gamma(0)}:T_{\gamma(t)}\M \to T_{\gamma(0)}\M$ is the parallel transport of $\nabla$ along $\gamma$ (see e.g., \cite[Sec.~4]{BBS55} or \cite[Sec.~10]{Wan67}).
Thus, as in the classical material science context, the Burgers vector is the sum of the tangents to the curve; in order to make sense of this on manifolds one has first to parallel transport all the tangent vectors to the same tangent space.

Burgers vectors are closely related to the {\em torsion tensor},
\[
T(X,Y) = \nabla_X Y - \nabla_Y X -[X,Y].
\]
The torsion $T$ is an infinitesimal Burgers vector in the following sense:
Let $p\in \M$ and let $\exp_p:T_p\M \to \M$ be a exponential map of $\nabla$.\footnote{Actually, any map $\phi:T_p\M\to \M$ with $\phi(0)=p$, whose differential at the origin is the identity will do.}
Let $\sigma_\e:[0,1]\to T_p\M$ be the parallelogram from the origin built from the vectors $\sqrt{\e}X, \sqrt{\e}Y$, and let $\gamma_\e = \exp_p(\sigma_\e)$ (see Figure~\ref{fig:torsion}).
Then
\[
\left.\frac{d}{d\e}\right|_{\e=0} \b_{\gamma_\e} = T(X,Y).
\]
This result is due to Cartan; see \cite[Chapter~III, Section~2]{Sch54} for a proof.

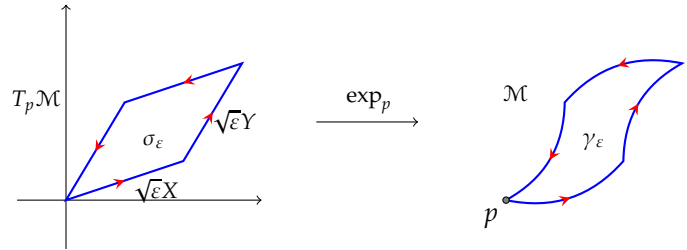
\begin{figure}
\begin{center}
\begin{tikzpicture}[scale=1.3]
	\draw[->] (-0.5,0) -- (2,0);
	\draw[->] (0,-0.5) -- (0,2);
	\tkzText(-0.3,1){ $T_p\M$}
	\path [thick, draw=blue,postaction={on each segment={mid arrow=red}}]
  	(0,0) -- (1.2,0.4)  -- (1.8,1.4)  -- (0.6,1) -- cycle;
 	\tkzText(0.9,0.6){ $\sigma_\e$}
 	\tkzText(0.9,0.1){ $\sqrt{\e}X$}
 	\tkzText(1.7,0.8){ $\sqrt{\e}Y$}
	
	\draw[->] (2.56,0.8) -- (3.6,0.8); 
	\tkzText(3.1,1){ $\exp_p$}
	
	\begin{scope}[xshift=4.5cm]
	\tkzText(0.1,1.1){ $\M$}
	\tkzDefPoint(0,0){p};
	\path [thick, draw=blue,postaction={on each segment={mid arrow=red}}]
	  (0,0) to [bend right] (1.2,0.4)  to [bend left] (1.8,1.4)  to [bend right] (0.6,1) to [bend left] cycle;
	\tkzDrawPoint(p);
	\tkzLabelPoint[below left](p){$p$};
	\tkzText(0.9,0.6){ $\gamma_\e$}
	\end{scope}

\end{tikzpicture}
\end{center}
\caption{The Burgers vector associated a loop $\gamma_\e$ in $\M$ which is the image under the exponential map of a parallelogram $\sigma_\e$ in $T_p\M$ with edges $\sqrt{\e}X$ and $\sqrt{\e}Y$  tends asymptotically to $\e\, T(X,Y)$. }
\label{fig:torsion}
\end{figure}

In the case of a body with finitely many dislocations $(\M,\g,\nablaLC)$ (according to Definition~\ref{def:body_w_dislocations}), the Burgers vector for any curve that does not encircle one of the dislocation cores is zero. This follows from the fact that every simply-connected submanifold of $\M$ which does not contain dislocations is isometrically embeddable into Euclidean plane, and that the Burgers vector of any closed curve in the plane is zero.
To quantify the Burgers vector associated with a curve encircling a single dislocation, consider the manifold depicted in Figure~\ref{fig:dislocation}.
One can then see that the magnitude of the Burgers vector is 
\beq
\label{eq:Burgers_vector_magnitude}
\textup{b} = 2d \sin(\theta/2),
\eeq
where $d$ is the length of the dislocation core (the distance between the two singular points forming the curvature dipole), and $\theta$ is the magnitude of the cone defect (see Figure~\ref{fig:Burgers}).
For a general Burgers circuit, the Burgers vector is the sum of the contributions of the dislocation cores it encircles (after parallel transporting each contribution to the base point).

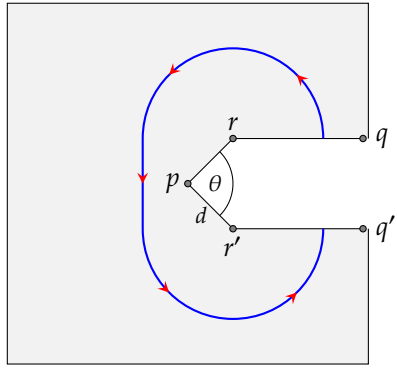
\begin{figure}
\begin{center}
\begin{tikzpicture}[scale=1.2]
	%\draw[black,fill=gray!10] (0,0) circle (2cm);
	\draw[black,fill=gray!10] (-2,-2) -- (2,-2) -- (2,2) -- (-2,2) -- cycle;
	\fill[white] ((2.1,0.5) -- (0.5,0.5) -- (0,0) -- (0.5,-0.5) -- (2.1,-0.5) -- cycle;
	\draw[black] ((1.94,0.5) -- (0.5,0.5) -- (0,0) -- (0.5,-0.5) -- (1.94,-0.5) ;
	\tkzDefPoint(1.94,-0.5){Cp};
	\tkzDefPoint(0.5,-0.5){Ep};
	\tkzDefPoint(0,0){F};
	\tkzDefPoint(0.5,0.5){E};
	\tkzDefPoint(1.94,0.5){C};
	\tkzDrawPoints(Cp,Ep,F,E,C);
	\tkzLabelPoint[left](F){$p$}
	\tkzLabelPoint[right=2pt](C){$q$}
	\tkzLabelPoint[right=2pt](Cp){$q'$}
	\tkzLabelPoint[above](E){$r$}
	\tkzLabelPoint[below](Ep){$r'$}
	\tkzDefPoint(2,0){Fp};
	%\tkzDrawSegment[dashed](F,Fp);
	\tkzMarkAngle[size=0.5cm](Ep,F,E);
	\tkzText(0.3,0.0){$\theta$}
	\tkzText(0.15,-0.35){$d$}
	\path [thick, draw=blue,postaction={on each segment={mid arrow=red}}]
  	(1.5,0.5)  arc (0:180:1.)
	(-0.5,0.5) -- (-0.5,-0.5)
	(-0.5,-0.5)  arc (180:360:1.);
\end{tikzpicture}
\end{center}
\caption{A Burgers circuit yielding a Burgers vector whose magnitude is $2d \sin(\theta/2)$, where $\theta$ is the disclination angle and $d$ is the distance between two disclinations forming the edge-dislocation. The vector points downwards from a chosen base-point of the circuit.}
\label{fig:Burgers}
\end{figure}

It follows that by changing $d$ and $\theta$ in Figure~\ref{fig:dislocation}, while keeping $\textup{b} = 2d \sin(\theta/2)$ fixed, we can obtain "the same" dislocation in different ways, in the sense that a Burgers circuit around the dislocation core will not be able to distinguish between the two.
Nevertheless, the choice of $d$ and $\theta$ will be important from the viewpoint of convergence of bodies with dislocations, as depicted in the sketch of the proof below.

%%%%%%%%%%%%%%%%%%%%%
\paragraph{Main result: convergence in the geometric paradigm.}
We now describe a version of the main theorem of \cite{KM15b}, stating that in the geometric paradigm, every two-dimensional body with distributed dislocations is a limit of bodies with finitely many dislocations.
\begin{theorem}[Homogenization of dislocations, geometric paradigm]
\label{thm:manifold_conv}
For every compact two-dimensional Weitzenb\"ock manifold $(\M,\g,\nabla)$ and parallel orthonormal frame $E$,
there exists a sequence of bodies with finitely-many dislocations $(\M_n,\g_n,\nablaLC_n)$ and parallel orthonormal frames $E_n$, such that there exist homeomorphisms $F_n:\M_n\to \M$, whose restrictions to the smooth part of $\M_n$ are smooth embeddings,  satisfying
\beq
\label{eq:frame_conv}
\| dF_n \circ E_n - E\|_{L^\infty} \to 0.
\eeq
\end{theorem}

Note that an orthonormal parallel frame $E$ contains all the geometric information of the Weitzenb\"ock manifold: 
since $E:\R^d\to T\M$ is orthonormal, it induces $\g$ by pushing forward the Euclidean metric on $\R^d$ (as in \eqref{eq:implants_define_metric}), and since it is parallel, it induced the parallel transport of $\nabla$ (see \eqref{eq:implants_define_connection}).
Therefore, the notion of convergence in Theorem~\ref{thm:manifold_conv}, which is defined through the convergence of orthonormal parallel frames, induces the convergence of the entire structure $(\M_n,\g_n,\nablaLC_n)\to (\M,\g,\nabla)$ of the Weitzenb\"ock manifolds.

We can also view Theorem~\ref{thm:manifold_conv} as a theorem about the convergence of manifolds endowed with frame fields $(\M_n,E_n)\to (\M,E)$, where each of the manifolds $(\M_n,E_n)$ induces the structure of a body with edge dislocations as in Definition~\ref{def:body_w_dislocations}.
This viewpoint, while maybe somewhat less natural from a geometric perspective, will be useful in the next section (convergence in the constitutive paradigm), when we associate these manifolds with a fixed archetype and consider $E_n$ and $E$ as implant maps.

%%%%%%%%%%%%%%%%%
\subsection{Sketch of proof of Theorem~\ref{thm:manifold_conv}}
Theorem~\ref{thm:manifold_conv} is an approximation result: given a manifold $(\M,\g,\nabla)$, we approximate it with a sequence of manifolds of a specific type (Definition~\ref{def:body_w_dislocations}).

\paragraph{Approximation by disclinations}
Before we describe the main idea of this approximation, it is illustrative to present a similar one, which is somewhat more intuitive---the approximation of a Riemannian surface by locally-flat surfaces with disclinations.
Given a surface $(\M,\g)$, we approximate it as follows:
\begin{enumerate}
\item First, assume that $\M$ does not have a boundary.
Take a geodesic triangulation of the manifold---a set of points in $\M$, connected by minimizing geodesics that do not intersect, such that the resulting partition $\M$ consists of geodesic triangles (such triangulations exist; see for example \cite[Note~3.4.5.3]{Ber02}).
If $\M$ has a boundary, triangulate a subdomain $\M'\subset \M$, such that the distance between $\pl\M'$ and $\pl\M$ is small (of the order of the distance between the vertices).
\item Construct a manifold by replacing each triangle with a Euclidean triangle with the same edge lengths.
Since $\M$ is (generally) not flat, the angles of the original geodesic triangles differ from the angles of their Euclidean counterparts (by the Gauss-Bonnet theorem, the angles of each geodesic triangle generally do not sum up to $\pi$).
\end{enumerate} 
This way we obtain a topological manifold which is smooth and flat everywhere but at the vertices, which are cone singularities (disclinations)---the angles around each vertex do not generally sum up to $2\pi$, since they differ from the angles of the original geodesic triangulation.
This approximation of the surface is similar to the approximation of a sphere by a football (soccer ball), using triangles rather than pentagons and hexagons.

By choosing finer and finer triangulations, say, triangulations in which the edge-lengths are of order $1/n$ for $n\gg 1$, it is clear (intuitively) that one obtains better and better approximations of the original manifold; they converge as metric spaces to the original manifold (see \cite{DVW15} for an explicit estimate) while the distribution-valued curvatures converge to the smooth curvature of $\g$ (see \cite{CMS84}).

%%%%%%%%%%%
\paragraph{The approximating sequence for Theorem~\ref{thm:manifold_conv}}
The idea behind the proof of Theorem~\ref{thm:manifold_conv} is very similar: Construct a fine geodesic triangulation of the Weitzenb\"ock manifold $(\M,\g,\nabla)$, and then replace each triangle with a locally-flat one to obtain a body with finitely many dislocations.
The difference between the two constructions is in the triangulation and in the type of  locally-flat replacements.
\begin{enumerate}
\item Take a triangulation of $(\M,\g,\nabla)$ in which the edges are $\nabla$-geodesics; those differ generally from the Levi-Civita geodesics and are not even locally length-minimizing. At the $n$th stage, we choose the triangulation such that the length of each edge is between (say) $1/n$ and $3/2n$, and all the angles are bounded between $\delta$ and $\pi -\delta$ for some $\delta>0$ independent of $n$ (to ensure that all the triangles are uniformly non-degenerate as $n\to \infty$).
The existence of a geodesic triangulation, based on a non-Levi-Civita connection, is not trivial; it is proved in \cite[Proposition~3.1]{KM15b}.
Denote the skeleton of this triangulation (the union of all the edges) by $X_n$.

\item Since the Gauss-Bonnet theorem holds for a metrically-consistent connection (see \cite[Theorem~B.1]{KM15b}), and since $\nabla$ is metrically-consistent and has zero curvature, the angles of each geodesic triangle sum up to $\pi$.
In other words, if a geodesic triangle has edge lengths $a,b,c$ and angles $\alpha,\beta,\gamma$, then $\alpha+\beta+\gamma=\pi$; 
the angles are however ``wrong" in the sense that generally $\alpha \ne \alpha_0$, $\beta\ne \beta_0$ and $\gamma \ne \gamma_0$, where $\alpha_0,\beta_0,\gamma_0$ are the angles of the Euclidean triangle having edge-lengths $a,b,c$. 
Since the geodesic triangles are uniformly regular, the angles do not deviate much from the angle of the Euclidean triangle,
\beq
\label{eq:angle_difference}
|\alpha - \alpha_0| , |\beta- \beta_0| , |\gamma-\gamma_0| = O(1/n).
\eeq
See \cite[Corollary~2.7]{KM15b}.\footnote{The estimate \eqref{eq:angle_difference} does not appear in this corollary explicitly;  it follows from its fourth part, using the fact a small triangle on $\M$ with edges that are Levi-Civita geodesics is, to leading order, Euclidean (this follows from standard triangle comparison results).}

\item As stated above, the Euclidean triangle having side lengths $a,b,c$ does not have angles $\alpha,\beta,\gamma$; 
	however, if Condition \eqref{eq:angle_difference} holds and $\alpha+\beta+\gamma=\pi$, then there exists a manifold containing a single dislocation (according to Definition~\ref{def:body_w_dislocations}), whose boundary is a triangle whose edge lengths and angles are $a,b,c$ and $\alpha,\beta,\gamma$ \cite[Proposition~3.3]{KM15b} (see Figure~\ref{fig:defective_triangle}).
The only additional parameter entering in this construction is the Burgers vector associated with the perimeter of the triangle, and whose magnitude is of order $O(1/n^2)$.
The precise location of the dislocation core inside the ``triangle" is arbitrary (as long as it does not intersect the boundary), 
as is the choice of the parameters $\theta$ and $d$ (see \eqref{eq:Burgers_vector_magnitude}).

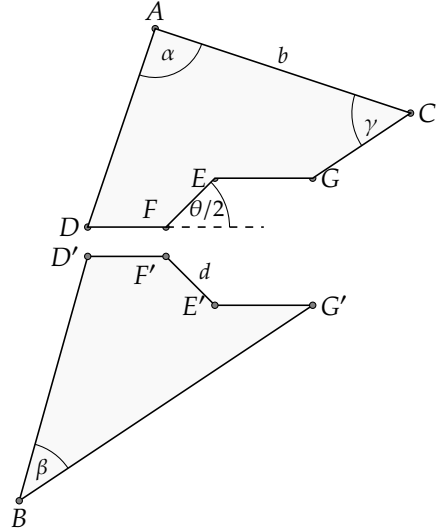
\begin{figure}
\begin{center}
\begin{tikzpicture}[scale=1.3]
	\tkzDefPoint(0,0){B};
	\tkzDefPoint(3,2){Cp};
	\tkzDefPoint(2.0,2){Ep};
	\tkzDefPoint(1.5,2.5){Fp};
	\tkzDefPoint(0.7,2.5){Dp};
	\tkzDrawPolygon[fill=gray!5](B,Cp,Ep,Fp,Dp);
	\tkzDrawPoints(B,Cp,Ep,Fp,Dp);
	\tkzDefPoint(0.7,2.8){D};
	\tkzDefPoint(1.5,2.8){F};
	\tkzDefPoint(2.5,2.8){FF};
	\tkzDefPoint(2,3.3){E};
	\tkzDefPoint(3,3.3){C};
	\tkzDefPoint(3+1,3.3+0.666){newC};
	\tkzDefPoint(2.7,3.3+0.666){tmp};
	\tkzDefPoint(1.19+0.2,4.25+0.580){A};
	\tkzDrawPoints(A,D,E,F,C,newC);
	\tkzDrawPolygon[fill=gray!5](A,newC,C,E,F,D);
	%\tkzLabelSegment[right=0.2](B,Cp){$a$}
	\tkzLabelSegment[above](A,newC){$b$}
	%\tkzDrawSegment[dashed](newC,tmp)
	\tkzLabelPoint[below](B){$B$}
	\tkzLabelPoint[above](A){$A$}
	\tkzLabelPoint[left](D){$D$}
	\tkzLabelPoint[left](Dp){$D'$}
	\tkzLabelPoint[right](C){$G$}
	\tkzLabelPoint[right](newC){$C$}
	\tkzLabelPoint[right](Cp){$G'$}
	\tkzLabelPoint[above left=0.01](F){$F$}
	\tkzLabelPoint[below left=0.01](Fp){$F'$}
	\tkzLabelPoint[left](E){$E$}
	\tkzLabelPoint[left](Ep){$E'$}
	\tkzLabelAngle[pos=0.3](D,A,newC){$\alpha$}
	\tkzMarkAngle[size=0.5](D,A,newC)
	\tkzLabelAngle[pos=0.4](Cp,B,Dp){$\beta$}
	\tkzMarkAngle[size=0.6](Cp,B,Dp)
	%\tkzLabelAngle[pos=0.4](A,newC,C){$\gamma$}
	\tkzText(3.6,3.82){ $\gamma$}
	\tkzMarkAngle[size=0.6](A,newC,C)
	%\tkzLabelAngle[pos=0.8](A,newC,tmp){$\vp$}
	%\tkzMarkAngle[size=1](A,newC,tmp)
	%\tkzText(2.4,1.87){{\small $\gamma-\vp$}}
	%\tkzMarkAngle[arc=lll,size=0.95](Ep,Cp,B)
	\tkzDrawSegment[dashed](F,FF);
	\tkzMarkAngle[size=0.65](FF,F,E)
	\tkzLabelAngle[pos=0.45](FF,F,E){{\small $\theta/2$}}
	\tkzText(1.9,2.35){ $d$}
\end{tikzpicture}
\end{center}
\caption{A triangle containing a single edge-dislocation. Given angles $\alpha,\beta,\gamma$ adding up to $\pi$ and edge lengths $a,b,c$, we construct a defective triangle by identifying the edges $DF$ and $D'F'$, $FE$ and $F'E'$, and $EG$ and $E'G'$, such that $CG+G'B=a$, $AC=b$ and $AD+D'B=c$.}
\label{fig:defective_triangle}
\end{figure}

\item The approximation of $(\M,\g,\nabla)$ is obtained by replacing each triangle in the triangulation with a ``dislocated" triangle having the same edge-lengths and angles.
Denote the resulting manifold by $(\M_n,\g_n,\nablaLC_n)$, and the skeleton of the triangulation on $\M_n$ by $Y_n$.
Since the angles in each triangle in $Y_n$ are the same as in the corresponding  triangle in $X_n$, it follows that the angles around each vertex in $Y_n$ sum up to $2\pi$.
In other words, there are no cone defects (disclinations) at the vertices of the triangulation; 
the only singularities in $\M_n$ are the dislocation cores within each triangle.
Hence, $(\M_n,\g_n,\nablaLC_n)$ is a body with finitely-many dislocations according to Definition~\ref{def:body_w_dislocations}
(see Figure~\ref{fig:approx})
\end{enumerate}

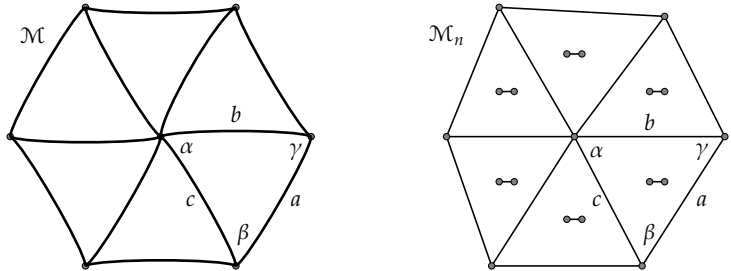
\begin{figure}
\begin{center}
\begin{tikzpicture}[scale=1]
	\tkzDefPoint(0,0){A};
	\tkzDefPoint(2,0){B};
	\tkzDefPoint(1,1.72){C};
	\tkzDefPoint(-1,1.72){D};
	\tkzDefPoint(-2,0){E};
	\tkzDefPoint(1,-1.72){F};
	\tkzDefPoint(-1,-1.72){G};
	\tkzDrawPoints(A,B,C,D,E,F,G);
	\draw [line width=1pt] (0,0) edge[bend left=30,distance=0.2cm] (2,0);
	\draw [line width=1pt] (0,0) edge[bend left=30,distance=0.3cm] (1,1.72);
	\draw [line width=1pt] (2,0) edge[bend left=30,distance=0.2cm] (1,1.72);
	\draw [line width=1pt] (1,1.72) edge[bend left=30,distance=0.2cm] (-1,1.72);
	\draw [line width=1pt] (0,0) edge[bend left=30,distance=0.2cm] (-1,1.72);
	\draw [line width=1pt] (0,0) edge[bend left=30,distance=0.2cm] (-2,0);
	\draw [line width=1pt] (-1,1.72) edge[bend right=30,distance=0.2cm] (-2,0);
	\draw [line width=1pt] (-1,-1.72) edge[bend right=30,distance=0.2cm] (-2,0);
	\draw [line width=1pt] (-1,-1.72) edge[bend right=30,distance=0.2cm] (0,0);
	\draw [line width=1pt] (1,-1.72) edge[bend right=30,distance=0.2cm] (0,0);
	\draw [line width=1pt] (1,-1.72) edge[bend right=30,distance=0.2cm] (2,0);
	\draw [line width=1pt] (1,-1.72) edge[bend right=30,distance=0.2cm] (-1,-1.72);
	\tkzText(-1.7,1.35){ $\M$}
	\tkzText(0.35,-0.15){ $\alpha$}
	\tkzText(1.1,-1.3){ $\beta$}
	\tkzText(1.8,-0.2){ $\gamma$}
	\tkzText(1,0.3){ $b$}
	\tkzText(1.8,-0.85){ $a$}
	\tkzText(0.4,-0.85){ $c$}

	\begin{scope}[xshift=5.5cm]
	\tkzDefPoint(0,0){A};
	\tkzDefPoint(2,0){B};
	\tkzDefPoint(1.2,1.6){C};
	\tkzDefPoint(-1,1.72){D};
	\tkzDefPoint(-1.7,0){E};
	\tkzDefPoint(-1.1,-1.72){F};
	\tkzDefPoint(0.9,-1.72){G};
	\tkzDrawSegments(A,B A,C A,D A,E A,F A,G B,C C,D D,E E,F F,G G,B);
	\tkzDrawPoints(A,B,C,D,E,F,G);
	\tkzDefPoint(1,0.6){B1};
	\tkzDefPoint(1.2,0.6){B2};
	\tkzDrawSegment(B1,B2);
	\tkzDrawPoints(B1,B2);
	\tkzDefPoint(-0.1,1.1){B1};
	\tkzDefPoint(0.1,1.1){B2};
	\tkzDrawSegment(B1,B2);
	\tkzDrawPoints(B1,B2);
	\tkzDefPoint(-1,0.6){B1};
	\tkzDefPoint(-0.8,0.6){B2};
	\tkzDrawSegment(B1,B2);
	\tkzDrawPoints(B1,B2);
	\tkzDefPoint(-1,-0.6){B1};
	\tkzDefPoint(-0.8,-0.6){B2};
	\tkzDrawSegment(B1,B2);
	\tkzDrawPoints(B1,B2);
	\tkzDefPoint(-0.1,-1.1){B1};
	\tkzDefPoint(0.1,-1.1){B2};
	\tkzDrawSegment(B1,B2);
	\tkzDrawPoints(B1,B2);
	\tkzDefPoint(1,-0.6){B1};
	\tkzDefPoint(1.2,-0.6){B2};
	\tkzDrawSegment(B1,B2);
	\tkzDrawPoints(B1,B2);
	\tkzText(-1.7,1.35){ $\M_n$}
	\tkzText(0.3,-0.2){ $\alpha$}
	\tkzText(1.,-1.3){ $\beta$}
	\tkzText(1.7,-0.2){ $\gamma$}
	\tkzText(1,0.2){ $b$}
	\tkzText(1.7,-0.85){ $a$}
	\tkzText(0.3,-0.85){ $c$}

	\end{scope}
\end{tikzpicture}
\end{center}
\caption{Approximating the smooth Weitzenb\"ock manifold $(\M,\g,\nabla)$ by manifolds $(\M_n,\g_n,\nablaLC_n)$ with singularities. Each $\nabla$-geodesic triangle in $(\M,\g,\nabla)$ is replaced by locally-Euclidean triangle, having the same angles and side lengths, and containing a single dislocation (the core of each dislocation is sketched here as a segment inside the triangle).}
\label{fig:approx}
\end{figure}

%%%%%%%%%%%%
\paragraph{Convergence of the approximating sequence}
The next step is to show that $(\M_n,\g_n,\nablaLC_n)$ converges to $(\M,\g,\nabla)$ in the sense of Theorem~\ref{thm:manifold_conv}.
That is, show that given a $\nabla$-parallel orthonormal frame $E$ on $\M$, there exist $\nablaLC_n$-parallel orthonormal frames $E_n$ on $\M_n$ and maps $F_n:\M_n\to \M$ such that \eqref{eq:frame_conv} holds.

Given $E$, the construction of $E_n$ is very natural:
let $\{e_1,e_2\}$ be the standard basis of $\R^2$. 
Let $p_n\in \M$ be a vertex in $X_n$, the $n$th triangulation of $\M$, and let $q_n$ be its corresponding vertex in $Y_n$.
Each of the vectors $E_{p_n}(e_1), E_{p_n}(e_2)$ is a $\g$-unit vector in $T_{p_n}\M$, which is uniquely defined by its angles with the $\nabla$-geodesics in $X_n$ emanating from $p_n$.
Define $(E_n)_{q_n}(e_i)$ to be the $\g_n$-unit vector in $T_{q_n}\M_n$ which forms the same angles with the corresponding geodesics emanating from $q_n$.
This defines $E_n$ everywhere by $\nablaLC_n$-parallel transport.
Note that this relation between $E_{p_n}(e_i)$ and $E_{q_n}(e_i)$ actually holds for {\em any} vertex $p_n\in X_n\subset \M$ and corresponding vertex $q_n\in Y_n\subset\M_n$. 
This follows from the construction, since $X_n$ consists of $\nabla$-geodesics and $Y_n$ consists of $\nabla_n$-geodesics, and the angles in the corresponding triangles match.

The construction of $F_n$ is more subtle.
Since $X_n$ and $Y_n$ have the same graph structure, and the lengths of its corresponding edges are the same, there is a natural map between these skeletons (the isometry of their graph metric);
it is natural to define the restriction of $F_n$ to $Y_n$ to be this map.  
Next, note that at every corresponding pair of vertices $p_n\in \M$, $q_n\in \M_n$, the frame fields induce an isometry $A:= E_{p_n}\circ ((E_n)_{q_n})^{-1}: T_{q_n}\M_n\to T_{p_n}\M$.
Define, $F_n$ in a neighborhood of $q_n$ by
\[
F_n (q) := \exp^\nabla_{p_n} \brk{ A \circ (\exp^{\nabla_n}_{q_n})^{-1} (q)}.
\]
By construction, this map respects the mapping of $Y_n$ to $X_n$, and moreover, $d_{q_n}F_n$ maps $(E_n)_{q_n}$ to $E_{p_n}$, and hence $|dF_n\circ E_n - E|$ is small near $p_n$.
In \cite[Section~4]{KM15b}, it is proved that $F_n$ can be extended in this way to a map that satisfies $|dF_n\circ E_n - E| = O(1/n)$ uniformly everywhere outside a small neighborhood, of diameter $o(1/n)$, of the dislocation core. Note that \cite{KM15b} aims at a slightly different notion of convergence (compared to Theorem~\ref{thm:manifold_conv}), hence this statement is not explicit in \cite{KM15b}, however the proof of Proposition~4.3 in \cite{KM15b} yields this result.

It remains to analyze the vicinity of a dislocation core.
Recall that in the construction of $\M_n$, only the Burgers vector inside each triangle was taken into account.
For understanding the behavior of $F_n$ near the dislocation core, and only there, the exact construction of the dislocation plays a role: in \cite{KM15b}, a dislocation of magnitude $O(1/n^2)$ is built using an arbitrary, but fixed, dislocation angle $\theta \approx 1$, whereas the size of the dislocation core is $d=O(1/n^2)$.
In this case, extensions of $F_n$ to the dislocation core only satisfy that $|dF_n\circ E_n - E|$ is bounded near the core (an explicit construction can be seen in \cite[Section~3.2]{KM15}).
This only yields $L^p$ convergence in \eqref{eq:frame_conv}, for any $p<\infty$, but not $L^\infty$, which is enough for the version of Theorem~\ref{thm:manifold_conv} that appear in \cite{KM15b}, but not to Theorem~\ref{thm:manifold_conv} as stated here (which is needed for the next section).
If however one takes $\theta = o(1)$ and $d=o(1/n)$ (such that the dislocation magnitude \eqref{eq:Burgers_vector_magnitude} is as prescribed), $F_n$ can be extended to the dislocation core such that \eqref{eq:frame_conv} holds.\footnote{In \cite[Section~3.2]{KM15}, choosing $\theta = o(1)$, $d=o(1/n)$ implies, in the notation of \cite{KM15}, $n^{-1} \ll D \ll 1$, which then implies $L^\infty$ convergence (see the proof of \cite[Proposition~2]{KM15}).
The general case is very similar, since we are only considering minuscule pieces of the manifolds, in which the only geometry that plays a role is the structure of the singular points (everything else is uniformly close to the trivial Euclidean plane).
See also \cite[Section~2.3.2, Example~2]{KM16}.}

%%%%%%%%%%%%%%%%%%%%%%%%%%%%%%%%%%%%%%%%%%%%%%%%%%%%%%%
\section{Homogenization of dislocations: constitutive paradigm}
\label{sec:homogenization}

Our aim in this section is to prove a homogenization theorem for dislocations within the constitutive paradigm, thus proving the third and final part of Theorem~\ref{thm:loop}.
To this end, some assumptions about the archetype $\calW:\R^d\times \R^d \to [0,\infty)$ are required:
\begin{enumerate}
\item Growth conditions:
\beq
\label{eq:growth_conditions}
\alpha(-1 + |A|^p) \le \calW(A) \le \beta(1+|A|^p),
\eeq
for some $p\in (1,\infty)$ and $\alpha,\beta>0$.
\item Quasiconvexity:\footnote{The quasiconvexity assumption is natural from a variational point of view, as it guarantees the existence of an energy minimizer of the functional; see also Remark~\ref{rmk:qc}.}
\[
\calW(A) \le \int_{(0,1)^d} \calW(A + d\vp(x))\, dx \qquad \text{for every } \vp\in C_c^\infty ((0,1)^d, \R^d).
\]
\item Solid symmetry group: $\mathcal{G}(\calW) \le \SO(d)$.
\end{enumerate}

%%%%%%%%%%%%%%%%%
\begin{remark}
It is usual to assume that $\calW$ is frame-indifferent and that $\calW(A) = 0$ iff $A\in \SO(d)$, but both assumptions are not required for the theorem.
Moreover, quasiconvexity and \eqref{eq:growth_conditions}, implies that $\calW$ satisfies the $p$-Lipschitz property \cite[Proposition~2.32]{Dac08}:
\beq
\label{eq:p_Lipschitz}
|\calW(A)-\calW(B)| \le C(1+|A|^{p-1}+|B|^{p-1})|A-B|,
\eeq
for some $C>0$ (and in particular $\calW$ is continuous).
\end{remark}

%%%%%%%%%%%%%%%%%
\begin{example}
We describe now two simple examples of archetypes $\calW$ satisfying the above hypotheses---one isotropic and one having a discrete symmetry group:
\begin{enumerate}
\item The isotropic archetype $\calW_{\text{iso}}(A) = \dist^p(A,\SO(d))$ (as in Example~\ref{ex:dist2}) satisfies all the hypotheses but for quasi-convexity. This can be rectified by replacing $\calW_{\text{iso}}$ with its quasiconvex envelope $Q\calW_{\text{iso}}$, which is an isotropic archetype satisfying all of the hypotheses.
In two dimensions, we can write $Q\calW_{\text{iso}}$ explicitly for every $p\ge 2$ \cite{Sil01,Dol12}:
\[
Q\calW_{\text{iso}}(A) = 
\begin{cases}
\dist^p(A,\SO(d)) & \mu_1 + \mu_2 \ge 1 \\
(1 - 2\det A)^{p/2} & \mu_1 + \mu_2 \le 1,
\end{cases}
\]
where $\mu_1\ge |\mu_2| \ge 0$ are the signed singular values of $A$ (i.e., if $\sigma_1\ge \sigma_2 \ge0$ are the singular values, $\mu_1=\sigma_1$ and $\mu_2 = (\textup{sgn} \det A)\sigma_2$).
In higher dimensions, $Q\calW_{\text{iso}}$ is not known explicitly, however it is known that 
\[
c \calW_{\text{iso}} \le Q\calW_{\text{iso}} \le \calW_{\text{iso}}
\]
for some constant $c>0$ (see \cite[Proposition~10]{KM18}).

\item An example of an archetype having a discrete symmetry group is
\[
\calW_{\text{cubic}} (A) = \sum_{i=1}^d \beta_i \brk{ |Ae_i| - 1 }^2,
\]
where $\beta_i >0$ are parameters and $\{e_i\}$ is the standard basis of $\R^d$.
This energy density penalizes stretching along each of the lattice directions $e_i$.
Once again, this function is not quasi-convex, and its quasi-convex envelope is given by \cite[Lemma~4.1]{LO15}
\[
Q\calW_{\text{cubic}} (A) = \sum_{i=1}^d \beta_i \brk{ |Ae_i| - 1 }_+^2,
\]
where for $f\in\R$, $f_+$ denotes the maximum between $f$ and zero.
While $Q\calW_{\text{cubic}}$ satisfies all the assumptions, it is somewhat non-physical. For example, it does not penalize for compression (this is due to the fact that that $\calW_{\text{cubic}}$ is invariant under orientation reversal).
By adding to $\calW$ penalization for volume change (as in \cite{KM18}) or simply by considering $Q\calW_{\text{cubic}} + Q\calW_{\text{iso}}$ one obtains an archetype satisfying all the hypotheses and having a discrete symmetry group.
\end{enumerate}
\end{example} 

%%%%%%%%%%%%%
\begin{remark}
The assumption $\calW<\infty$ excludes physically-relevant archetypes in which $\calW(A)$ diverges as $A$ becomes singular (see, e.g., \cite[Theorem~4.10-2]{Cia88}).
The requirement $\calW<\infty$ is due to purely technical reasons that commonly appear in $\Gamma$-convergence results in elasticity when the elastic energy is $O(1)$.
\end{remark}
%%%%%%%%%%%%%

In the rest of this section, it is easier to consider that $\M$ is endowed with an orthonormal parallel frame field $E$ rather than a flat connection $\nabla$; as stated above, this is completely equivalent modulo a global rotation of $E$.

\begin{definition}
\label{def:I_tildeI}
Let $\calW$ be an archetype satisfying the above conditions.
Let $(\M,\g,E)$ be a Riemannian manifold with an orthonormal frame field $E$.
The elastic energy associated with $(\M,\g,E)$ and $\calW$ is
\[
I(f) = \int_\M \calW(df\circ E)\, \dVolg \qquad \text{$f\in W^{1,p}(\M;\R^d)$}.
\]
Note that $\g$ is an intrinsic metric for this energy, and that the connection $\nabla$, defined by declaring $E$ parallel, is a material connection.

As standard in these type of problems, we extend $I$ to $L^p(\M;\R^d)$ by 
\[
\tI(f) = 
\begin{cases}
\int_\M \calW(df\circ E)\, \dVolg  	& f\in W^{1,p}(\M;\R^d)\\
+\infty						& f\in L^p(\M;\R^d)\setminus W^{1,p}(\M;\R^d).
\end{cases}
\]
\end{definition}

%%%%%%%%

In order to define convergence of the energy functionals, each defined on a different manifold $\M_n$, we need a notion of convergence of maps $f_n:\M_n\to \R^d$:

\begin{definition}
\label{def:convergence_of_maps}
$(\M,\g)$ be a Riemannian manifold, and let $\M_n$ be topological manifolds.
Let $F_n:\M_n\to \M$ be homeomorphisms.
We say that a sequence of maps $f_n:\M_n\to \R^d$ converges to a map $f:\M\to \R^d$ in $L^p$ if
\[
\|f_n\circ F_n^{-1} - f\|_{L^p(\M;\R^d)}\to 0.
\]
\end{definition}

%%%%%%%%%%%%%
\begin{theorem}[$\Gamma$-convergence of elastic energies]
\label{thm:Gamma_conv}
Let $\calW$ be an archetype satisfying the above assumptions.
Let $(\M,\g,E)$, $(\M_n,\g_n,E_n)$ be Riemannian manifolds with orthonormal frames.
Let $\tI$, $\tI_n$ be their associated elastic energies according to Definition~\ref{def:I_tildeI}.
If there exists Lipschitz homeomorphisms $F_n:\M_n\to \M$ such that
\beq
\label{eq:E_n_to_E}
\| dF_n \circ E_n - E\|_{L^\infty} \to 0,
\eeq
then $\tI_n\to \tI$ in the sense of $\Gamma$-convergence, relative to the convergence induced by $F_n$, as defined in Definition~\ref{def:convergence_of_maps} (note that for Lipschitz maps, $dF_n\in L^\infty(T\M_n,F_n^*T\M)$, hence the convergence is well-defined).
\end{theorem}

\begin{remark}
\label{rmk:qc}
If $\calW$ is not quasiconvex (but \eqref{eq:p_Lipschitz} holds),
then it follows from slight changes in the proof below that $\tI_n$ converges to the functional associated with $(\M,\g,E)$ and the archetype $Q\calW$, which is the quasiconvex envelope of $\calW$.
Note that it is still true that $\g$ is an intrinsic metric and that $\nabla$ is a material connection, hence Figure~\ref{fig:loop} still holds.
\end{remark}

Combining Theorem~\ref{thm:manifold_conv} and Theorem~\ref{thm:Gamma_conv}, we conclude the proof of Theorem~\ref{thm:loop}:
\begin{corollary}
\label{cor:loop}
Every two-dimensional body with a continuous distribution of dislocations $(\M,\g,E)$ is a limit of bodies with finitely many dislocations $(\M_n,\g_n,E_n)$ in the sense of Theorem~\ref{thm:manifold_conv} (equivalently~\eqref{eq:E_n_to_E}).
Given an archetype $\calW$, the elastic energies associated with $(\M_n,\g_n,E_n)$ according to Definition~\ref{def:I_tildeI} $\Gamma$-converge to the elastic energy associated with $(\M,\g,E)$.
\end{corollary}

\begin{remark}
Note that we do not rescale the elastic energies of the bodies with dislocations, that is, we are considering energies that are of order $1$.
This fits the typical heuristics for energies of dislocations: that a dislocation with a Burgers vector of magnitude $\e$ will have a self energy (or core energy) of order $\e^2 \log |\e|$, and that the interaction energy between two such dislocations will be of order $\e^2$ (see, e.g., \cite{CL05,GLP10}, which treats this in a linear case where $\e^2$ is factored out).
Indeed, in our case $(\M_n,\g_n,E_n)$ contains an order of $n^2$ dislocations of order $\e \approx n^{-2}$, so the self energy is of order $n^2 \cdot \e^2 \log |\e| \to 0$, while the interaction energy is of order $n^4 \cdot \e^2 \approx 1$.
To the best of our knowledge, this is the first rigorous framework in which an order $1$ energy limit of bodies of dislocations is obtained in non-linear settings. 

Note also that for coercive archetypes, that is, archetypes that  satisfy $\calW_{\text{iso}}(A) \ge c \dist^p(A,\SO(d))$ for some $c>0$, the limiting energy associated with $(\M,\g,E)$ is bounded away from zero if $\g$ is non-flat, that is, there are no stress-free configurations.
\end{remark}

%%%%%%%%%%%%%
\subsection{Proof of Theorem~\ref{thm:Gamma_conv}}

Let $\tI_\infty$ be the $\Gamma$-limit of a (not-relabeled) subsequence of $\tI_n$. Such a subsequence always exists by the general compactness theorem of $\Gamma$-convergence (see Theorem~8.5 in \cite{Dal93} for the classical result, or Theorem 4.7 in \cite{KS08} for the case where each functional is defined on a different space).
It is enough to prove that $\tI_\infty =  \tI$.
Indeed, since by the compactness theorem, every sequence has a $\Gamma$-converging subsequence, the Urysohn property of $\Gamma$-convergence (see Proposition 8.3 in \cite{Dal93}) implies that if all converging subsequences converge to the same limit, then the entire sequence converges to that limit.

From \eqref{eq:E_n_to_E} it follows that 
\begin{enumerate}
\item $dF_n$ and $dF_n^{-1}$ are uniformly bounded.
\item  $(F_n)_\star\g_n\to \g$ in $L^\infty$, and in particular, $(F_n)_\star \dVolgn \to \dVolg$ in $L^\infty$.
\end{enumerate}

\begin{lemma}[Infinity case]
\label{lem:infinity_case}
Let $f\in L^p(\M;\R^d)\setminus W^{1,p}(\M;\R^d)$. Then,
\[
\tI_\infty(f) = \infty = \tI(f).
\]
\end{lemma}

\begin{proof}\smartqed
Suppose, by contradiction, that $\tI_\infty(f)<\infty$. Let $f_n\to f$ be a recovery sequence, namely,
\[
\lim_{n\to\infty} \tI_n(f_n) =  I_\infty(f) < \infty.
\] 
Without loss of generality we may assume that $\tI_n(f_n)<\infty$ for all $n$, and in particular,
$f_n\in W^{1,p}(\M_n,\R^d)$. 
The coercivity of $W_n$ implies that
\[
\sup_n \int_{\M_n}  |df_n|_{\g_n,\euc}^p \,\dVolgn < \infty.
\]
Thus, $f_n$ is uniformly-bounded in $W^{1,p}$, and since $dF_n^{-1}$ are uniformly-bounded,
$f_n\circ F^{-1}_n$ is also uniformly-bounded in $W^{1,p}(\M;\R^d)$, hence weakly converges (modulo a subsequence). 
By the uniqueness of the limit, this limit is $f$, hence $f\in W^{1,p}(\M;\R^d)$, which is a contradiction.
\qed\end{proof}

%%%
\begin{lemma}[Upper bound]
\label{lem:upper_bound}
For every $f\in W^{1,p}(\M;\R^d)$,
\[
\tI_\infty(f) \le  \tI(f).
\]
\end{lemma}

\begin{proof}\smartqed
Let $f\in W^{1,p}(\M;\R^d)$. 
Define $f_n = f\circ F_n \in W^{1,p}(\M_n;\R^d)$. Trivially, $f_n \to f$ in $L^p$ according to Definition~\ref{def:convergence_of_maps}
and by the definition of the $\Gamma$-limit,
\[
\tI_\infty(f) \le \liminf_n \tI_n(f_n). 
\]
It follows from the uniform convergence $dF_n\circ E_n \to E$ and $(F_n)_\star \dVolgn\to \dVolg$, using the $p$-Lipschitz property \eqref{eq:p_Lipschitz}, that
\[
\lim_n \tI_n(f_n) = \tI(f),
\]
that is
\beq
\label{eq:upper_bound}
\lim_n \int_{\M_n} \calW(df\circ dF_n \circ E_n)\,\dVolgn = \int_{\M} \calW(df\circ E) \,\dVolg.
\eeq
\qed\end{proof}

%%%
\begin{lemma}[Lower bound]
\label{lem:lower_bound}
For every $f\in W^{1,p}(\M;\R^d)$,
\[
\tI_\infty(f) \ge  \tI(f).
\]
\end{lemma}

\begin{proof}\smartqed
Let $f\in W^{1,p}(\M;\R^d)$, and let $f_n\in L^p(\M;\R^d)$ be a recovery sequence for $f$, that is
$f_n\circ F_n^{-1} \to f$ in $L^p$ and $\tI_n(f_n) \to \tI_\infty(f)$.
In particular, it follows that we can assume without loss of generality that $f_n\in W^{1,p}$, and that $f_n$ are uniformly bounded in $W^{1,p}$.
Therefore, $f_n\circ F_n^{-1} \rightharpoonup f$ in $W^{1,p}(\M;\R^d)$.
We need to show that
\beq
\label{eq:lower_bound}
\lim_n \tI_n(f_n) \ge \tI(f).
\eeq
Note that since $f\in W^{1,p}(\M;\R^d)$ and $f_n \in W^{1,p}(\M_n;\R^d)$, $\tI(f) = I(f)$ and $\tI_n(f_n) = I_n(f_n)$.
Since $dF_n\circ E_n \to E$ and $(F_n)_\star \dVolgn\to \dVolg$ uniformly, and $df_n\circ dF^{-1}_n$ are uniformly bounded in $L^p$, the $p$-Lipschitz property \eqref{eq:p_Lipschitz} implies that
\beq
\label{eq:auxiliary_calc_lower_bound}
\begin{split}
\lim_n I_n(f_n) &= \lim_n \int_{\M_n} \calW(df_n \circ E_n)\,\dVolgn \\
	&= \lim_n \int_{\M_n} \calW(df_n \circ dF_n^{-1} \circ E)\,\dVolg = \lim_n I(f_n\circ F_n^{-1}).
\end{split}
\eeq
Since $\calW$ is quasiconvex and satisfies \eqref{eq:growth_conditions},  $I(\cdot)$ is lower semicontinuous with respect to the weak topology of $W^{1,p}(\M;\R^d)$ \cite[Theorem~8.11]{Dac08}.
Since $f_n\circ F_n^{-1}$ converges weakly to $f$ in $W^{1,p}(\M;\R^d)$, 
\[
\lim_n I(f_n\circ F_n^{-1}) \ge I(f),
\]
which together with \eqref{eq:auxiliary_calc_lower_bound} implies \eqref{eq:lower_bound}.
\qed\end{proof}

%%%%%%%%%%%%%%%%%%%%%%%%%%%%%%%%%%%%%%%%%%%%%%%%%%%%%%%
\appendix
\section{The role of torsion in the equilibrium equations}
\label{sec:appendix}

In this section we analyze explicitly the equilibrium equations for a hyperelastic solid body having a continuous distribution of dislocations, and in particular, we address the role of torsion.
We will explain why torsion does not enter explicitly in the equilibrium of an isotropic body.
Similar equations are derived in \cite[Section~12]{Wan67} (without the hyperelasticity assumption).
Throughout this section we use the Einstein summation convention.

Let $\calW\in C^2(\R^d\times \R^d)$ be a solid undistorted archetype, and let $(\M,W)$ be a uniform solid material having $\calW$ as an archetype with respect to an implant map $E=\{E_p\}_{p\in \M}$. 
We denote the (matrix) argument of $\calW$ by $B=(B_1\, | \,\ldots \,| \, B_d)$, and by $\pl \calW/\pl B_i : \R^d\times \R^d\to \R^d$ the derivative of $\calW$ with respect to the column $B_i$ (this is a vector).

The implant map $E$ is a parallel frame of a flat material connection $\nabla$ (defined by \eqref{eq:implants_define_connection}) and it defines a metric $\g$ via \eqref{eq:implants_define_metric}.
$E$ is a $d$-tuple of vector fields which we denote by $E_1,\ldots,E_d$. 
Its co-frame $E^1,\ldots,E^d$ is the $d$-tuple of one-forms defined by $E^i(E_j) = \delta^i_j$.
The torsion tensor of $\nabla$ is given by
\[
T(E_i,E_j) = - [E_i,E_j] =: T_{ij}^k E_k,
\]
as follows from the definition of the torsion tensor $T(X,Y) = \nabla_X Y - \nabla_Y X - [X,Y]$, since $E_i$ are parallel, which means $\nabla E_i=0$.

The elastic energy functional corresponding to this elastic body is 
\[
I(f) = \int_\M W(df)\, \dVolg = \int_\M \calW(df\circ E) \, E^1\wedge \ldots \wedge E^d,
\]
defined on functions $f:\M\to \R^d$.
The Euler-Lagrange equations corresponding to this functional are, in a weak formulation,
\[
\int_\M \frac{\pl \calW}{\pl B_i}(df\circ E) \cdot E_i(h) \,\dVolg = 0\, \qquad \forall \, h\in C_c^\infty(\M;\R^d).
\]
where $E_i(h) = dh(E_i) : \M \to \R^d$, and $\cdot$ is the standard inner product in $\R^d$. 
The strong formulation of the Euler-Lagrange equations is 
\[
E_i\brk{\frac{\pl \calW}{\pl B_i}(df\circ E)} + \frac{\pl \calW}{\pl B_i}(df\circ E)\, \div E_i =0,
\]
or more explicitly,
\[
\frac{\pl^2 \calW}{\pl B_i\pl B_j}(df\circ E)\, E_i E_j (f)+ \frac{\pl \calW}{\pl B_i}(df\circ E)\, \div E_i =0,
\]
where $\div E_i$ is defined by the relation
\[
d(\iota_{E_i} \dVolg) = \div E_i \, \dVolg,
\]
where $\iota$ is the contraction operator.
Using the fact that $\dVolg = E^1\wedge \ldots \wedge E^d$, 
\[
\iota_{E_i} \dVolg = (-1)^{i+1} E^1\wedge \ldots \wedge E^{i-1} \wedge E^{i+1} \wedge\ldots\wedge E^d,
\]
hence
\[
\begin{split}
d(\iota_{E_i} \dVolg) &= (-1)^{i+1}  \left(dE^1\wedge \ldots \wedge E^{i-1} \wedge E^{i+1} \wedge\ldots\wedge E^d + \ldots \right.\\
	&\qquad \left. \ldots + (-1)^{d-1} E^1\wedge \ldots \wedge E^{i-1} \wedge E^{i+1} \wedge\ldots\wedge dE^d \right).
\end{split}
\]
By the definition of the exterior derivative, and the fact that $E^k(E_i) = \delta_i^k$,  
\[
dE^k(E_i,E_j) = E_i (E^k(E_j)) - E_j (E^k(E_i)) - E^k([E_i,E_j]) = T_{ij}^l E^k(E_l) = T_{ij}^k
\]
and therefore $dE^k = T_{ij}^k E^i\wedge E^j$, so $d(\iota_{E_i} \dVolg)$ simplifies to
\[
d(\iota_{E_i} \dVolg) = -T^j_{ji} \,\dVolg,
\]
hence $\div E_i = -T^j_{ji}$. It follows that the Euler-Lagrange equations are
\[
\frac{\pl^2 \calW}{\pl B_i\pl B_j}(df\circ E)\, E_i E_j (f)  - T_{ji}^j \frac{\pl\calW}{\pl A_i}(df\circ E) = 0.
\]
The trace of the torsion appears explicitly in the equations,
however, the torsion also appears, more implicitly, as the antisymmetric part $E_i E_j - E_j E_i = T_{ij}^k E_k$ of the first addend.

If the solid is isotropic, then the equilibrium equations are independent of the torsion.  Isotropy means that 
\[
\calW{B\circ R} = \calW{B} \qquad\text{for any $R\in \SO(d)$}.
\]
Using polar decomposition, this implies that there exists a function $\widetilde{\calW}:\operatorname{Sym}_+(d)\to\R$, where $\operatorname{Sym}_+(d)$ is the set of positive-semidefinite $d\times d$ symmetric matrices, such that 
\[
\calW(B) = \widetilde{\calW}(BB^T) 
\]
\cite[Theorem~3.4-1]{Cia88} (if one allows $B$ to be orientation reversing, then $\widetilde{\calW}$ also depends on the orientation of $B$, but this does not affect the argument below and therefore we ignore this subtlety).
It follows that
\[
I(f) = \int_\M W(df)\, \dVolg = \int_\M \calW(df\circ E) \, \dVolg = \int_\M \widetilde{\calW}((df\circ E)(df \circ E)^T) \, \dVolg.
\]
Choosing coordinates on $\M$, we can think of $df$ and $E$ as matrices.
In this case, since $E$ is an orthonormal frame for $\g$, $EE^T = \g^*$, the $\g$-metric on $T^*\M$ (whose coordinate are $\g^{ij}$).
Therefore, in coordinates,
\[
I(f) = \int_\M \widetilde{\calW}(df_x\circ \g_x^* \circ df_x^T) \, \sqrt{|\g|}(x) \, dx.
\]
In a more abstract language, 
\[
I(f) = \int_\M \tilde{\calW}(f_\star \g^*)\, \dVolg
\]
where $f_\star \g^*$ is the push-forward by $f$ of the metric $\g^*$ from $T^*\M$ to $\R^d$.
Either way, it is clearly seen that the energy (and therefore the equilibrium equations) only depend on $\g$ and not on the frame $E$, and therefore not on the connection $\nabla$ and its torsion which are derived from $E$.
%%%%%%%%%%%%%%%%%%%%%%%%%%%%%%%%%%%%%%%%%%%%%%%%%%%%%%%

\begin{acknowledgement}
This project was initiated in the Oberwolfach meeting "Material Theories" in July 2018. 
RK was partially funded by the Israel Science Foundation (Grant No. 1035/17), and by a grant from the Ministry of Science, Technology and Space, Israel and the Russian Foundation for Basic Research, the Russian Federation.
\end{acknowledgement}

%%%%%%%%%%%%%%%%%%%%%%%%%%%%%%%%%%%%%%%%%%%%%%%%%%%%%%%
\bibliographystyle{amsalpha}
\providecommand{\bysame}{\leavevmode\hbox to3em{\hrulefill}\thinspace}
\providecommand{\MR}{\relax\ifhmode\unskip\space\fi MR }
% \MRhref is called by the amsart/book/proc definition of \MR.
\providecommand{\MRhref}[2]{%
  \href{http://www.ams.org/mathscinet-getitem?mr=#1}{#2}
}
\providecommand{\href}[2]{#2}

\end{document}